\documentclass[11pt]{article}

\usepackage[margin=1in]{geometry}
\usepackage{amsmath,amssymb,amsthm}
\usepackage{graphicx}
\usepackage[plain,noend]{algorithm2e}
\usepackage{authblk}
\usepackage{natbib}
\usepackage{xcolor}
\usepackage{adjustbox}
\usepackage[section]{placeins}
\usepackage{hyperref}
\hypersetup{
  colorlinks=true,
  citecolor=blue,
  linkcolor=blue,
  urlcolor=blue,
  breaklinks=true
}

\graphicspath{{./figures/}}

\def\T{{ \mathrm{\scriptscriptstyle T} }}
\DeclareMathOperator{\tr}{tr}
\DeclareMathOperator{\diag}{diag}
\newcommand{\R}{\mathbb{R}}
\newcommand{\Le}{L_{\varepsilon}}
\newcommand{\bL}{L}
\newcommand{\Qw}{Q_{w}}
\newcommand{\Qpost}{Q_{\mathrm{post}}}
\newcommand{\GV}{\mathcal{G}_{V}}
\newcommand{\Ocal}{\mathcal{O}}
\newcommand{\neff}{n_{\mathrm{eff}}}

\newcommand{\tbl}[2]{\caption{#1}\begin{adjustbox}{max width=\textwidth}#2\end{adjustbox}}
\newenvironment{tabnote}{\par\vspace{4pt}\begingroup\footnotesize\noindent}{\par\endgroup}
\newenvironment{keywords}{\par\medskip\noindent\textbf{Keywords: }}{\par}

\theoremstyle{plain}
\newtheorem{theorem}{Theorem}
\newtheorem{lemma}{Lemma}
\newtheorem{corollary}{Corollary}
\newtheorem{proposition}{Proposition}
\theoremstyle{remark}
\newtheorem{remark}{Remark}

\title{Bigraphical Mat\'ern--Whittle (BMW) Processes for Fast Inference of
Big Multivariate Spatial Data on General Domains}

\author[1]{Debangan Dey\thanks{Corresponding author. Email: \texttt{debangan@tamu.edu}}}
\author[1]{Alokesh Manna}
\author[2]{Christopher J. Geoga}
\affil[1]{Department of Statistics, Texas A\&M University, College Station, Texas 77843, U.S.A.}
\affil[2]{Department of Statistics, University of Wisconsin--Madison, Madison, Wisconsin 53706, U.S.A.}
\date{}

\begin{document}

\renewcommand{\topfraction}{0.9}
\renewcommand{\bottomfraction}{0.85}
\renewcommand{\textfraction}{0.07}
\renewcommand{\floatpagefraction}{0.75}
\setcounter{topnumber}{4}
\setcounter{bottomnumber}{4}
\setcounter{totalnumber}{8}

\maketitle

\begin{abstract}
Large spatial data sets now record many correlated variables at many thousands of locations,
often on domains where Euclidean distance misrepresents proximity. The central difficulty is
modelling the cross-variable dependence jointly while retaining variable-level interpretation.
We introduce the bigraphical Mat\'ern--Whittle process, a multivariate Gaussian process that
resolves this with two graphs. A spatial graph generates the Mat\'ern structure of each
variable through a fractional power of a graph Laplacian, so the process is valid on any
topology, with per-variable range, smoothness and amplitude. A directed acyclic variable
graph encodes the scientific structure: we prove that each absent edge yields an exact
conditional independence between the corresponding fields. We further prove that the operator
determinant does not involve the cross-dependence coefficients, which keeps matrix-free
likelihood evaluation and Bayesian learning of the variable graph tractable at scale.
Estimation requires only sparse matrix--vector products and scales to tens of millions of
space--variable pairs. In simulations the method recovered parameters and graphs
accurately, remained robust under misspecification, and halved held-out prediction error on a
non-convex domain. In a spatial transcriptomics section with $19{,}809$ cells and $1{,}122$
genes, fitted in $75$ minutes on a laptop, borrowing across the learned gene graph reduced
held-out prediction error by $50$ to $91$ percent. Theoretical challenges, such as the achievable efficiency of estimating the
variance of the nugget, are also explored.
\end{abstract}

\begin{keywords}
Conditional independence; Gaussian process; Graphical model; Krylov method; Mat\'ern
covariance; Multivariate spatial data; Stochastic partial differential equation.
\end{keywords}

\section{Introduction}

Modern environmental and biomedical monitoring yields many correlated variables at thousands
to millions of spatial locations \citep{cressie2011statistics,banerjee2014hierarchical}. Examples include satellite retrievals of coupled ocean fields at millions
of pixels, multi-pollutant air-quality surfaces, imaging markers on irregular anatomical
domains, and spatially resolved transcriptomics, in which the expression of hundreds of genes
is mapped across many thousands of tissue locations
\citep{stahl2016visualization,chen2015spatially,moses2022museum}. Three demands collide in such data. The first is scale: direct likelihood evaluation
for a Gaussian process on $n$ locations and $p$ variables costs $O\{(np)^3\}$ operations,
which is prohibitive when $np$ reaches the millions. The second is multivariate structure:
scientific interest centres on which variables depend on which, so the cross-covariance should
carry an explicit conditional-independence graph. The third is geometry: coastlines, river
networks and cortical surfaces make Euclidean distance the wrong notion of proximity, and
a covariance function valid in one metric need not be valid in another.

Existing model classes meet these demands separately. The multivariate Mat\'ern model of
\citet{gneiting2010matern}, extended by \citet{apanasovich2010cross} and
\citet{apanasovich2012valid} and reviewed by \citet{genton2015cross}, is interpretable, as is
the linear model of coregionalization
\citep{wackernagel2003multivariate,gelfand2004nonstationary}. Both, however, presuppose a
valid global distance, scale poorly in both the number of locations and the number of
variables, and make every pair of variables dependent by
construction. The graphical Gaussian process of \citet{dey2022graphical} introduces
process-level conditional independence by stitching univariate Gaussian processes along a
decomposable variable graph. It scales far better in the number of variables, but its
per-clique likelihood still involves dense covariance operations of order $n^3$, so it does not
scale in the number of locations. Moreover, the graph must be known in advance and be
decomposable, and the component cross-covariances still require a Euclidean or geodesic metric. In spatial transcriptomics, Bayesian graphical models with covariate-dependent structure
recover gene networks that vary over the tissue \citep{ni2019bayesian,dawn2025spatially}.
These treat space as a covariate that modulates the network rather than modelling the
spatial and cross-gene dependence jointly, and computation at the full variable-location
scale is not their aim.

Scalable computation for large $n$ has been pursued through low-rank, tapered, composite,
nearest-neighbour and multi-resolution constructions
\citep{banerjee2008gaussian,furrer2006covariance,vecchia1988estimation,datta2016hierarchical,katzfuss2021general,peruzzi2022highly};
\citet{heaton2019case} compare many of these approaches, and variational Bayes further
accelerates the nearest-neighbour family \citep{song2025fast}. For large $p$, factor models,
spatial multivariate trees, mixtures of directed acyclic graphs and multivariate Vecchia
approximations represent the state of the art
\citep{zhang2021high,peruzzi2022spatial,jin2024bag,fahmy2022vecchia}. All of these
constructions approximate the model or the likelihood, and none carries an explicit
conditional-independence graph over the variables. Non-convex domains have their own
literature, from the horseshoe test domain of
\citet{ramsay2002spline} to barrier and network models
\citep{wood2008soap,bakka2019non,niu2019intrinsic,verhoef2010moving}; each of these
constructions is tied to its particular geometry, and none is multivariate.

The process model proposed here addresses the three demands within a single construction. The
variable graph is learned from the data, it may be an arbitrary directed acyclic graph, and
spatial dependence is defined by an operator on a graph rather than by a distance.
The construction descends from the stochastic partial differential equation approach.
\citet{whittle1954stationary} identified the Mat\'ern covariance as the stationary solution of
a fractional stochastic partial differential equation. \citet{lindgren2011explicit} made this
observation computational by discretizing the equation on a finite-element mesh, which yields
sparse-precision Markov representations; see \citet{lindgren2022spde} for a retrospective.
Systems of coupled equations generate multivariate fields
\citep{hu2013multivariate,hu2016systems,bolin2020multivariate}, and rational approximations
free the smoothness from integer constraints \citep{bolin2020rational}. We adopt the same
operator, but replace the mesh by the data graph and treat the smoothness as a continuous
parameter.

A parallel literature builds Mat\'ern-type processes directly from graph Laplacians, which
converge spectrally to the Laplace--Beltrami operator of an underlying manifold
\citep{coifman2006diffusion,calder2022improved}. Contributions include Mat\'ern Gaussian
processes on graphs and manifolds \citep{borovitskiy2020riemannian,borovitskiy2021matern},
graph representations of Mat\'ern fields \citep{sanzalonso2022spde}, restricted-domain
regression \citep{dunson2022graph}, and Whittle--Mat\'ern fields on metric graphs
\citep{bolin2024metric,bolin2025inference}. All of these processes are univariate, with the
graph as domain. The process proposed here appears to be the first multivariate member of the
family. It carries a second, learned graph over the variables: the spatial graph plays the
role of the domain, and the variable graph plays the role of the cross-covariance.

The computational route also differs in kind from the approximation literature. No
conditioning sets, knots, partitions or sparse factorizations are selected. The exact
likelihood and a variational surrogate are evaluated by iterative numerical linear algebra,
in the spirit of \citet{aune2014parameter} and \citet{gardner2018gpytorch}, and the
per-iteration cost is linear in $np$. The model itself is never approximated. Respect for
barriers is automatic, because connectivity rather than distance defines the operator.

This paper proposes the \emph{bigraphical Mat\'ern--Whittle} process, which is a Gaussian
process model specified through the transport (or equivalent covariance/precision) operator.
By taking this approach, modelling dependence when the relevant notion of distance is
non-Euclidean is easy and direct. Further, the approach empowers practitioners to
\emph{discover} dependence in highly multivariate processes, which can be challenging to do
in existing scalable GP frameworks. Thanks to a careful numerical design, the model can also
be handily applied to very large modern datasets. In the following sections, we describe a
careful specification, provide some theoretical discussion of subtle phenomena that occur in
estimation, and close by providing synthetic and real-data demonstrations of the method's
efficacy.

\section{The bigraphical Mat\'ern--Whittle process}\label{sec:model}

\subsection{Notation and data setting}\label{sec:notation}

We observe $p$ spatially indexed variables at locations $s_1,\ldots,s_n$ in a domain $D$. The
domain may be a subset of $\R^d$, a manifold, or an abstract network. Write $N = np$. The
observation of variable $j$ at location $i$ is $y_j(s_i)$, available for a subset
$\Ocal \subseteq \{1,\ldots,p\} \times \{1,\ldots,n\}$ of pairs with $M = |\Ocal|$. As is
standard for noisy spatial data, each observation is modelled as a smooth latent process plus
measurement error (\S\ref{sec:system}). The vector of latent processes, called the latent
field, is $w = (w_1^{\T},\ldots,w_p^{\T})^{\T} \in \R^{N}$ with $w_j \in \R^n$. Throughout,
$\Le \in \R^{n\times n}$ is a fixed sparse symmetric positive semi-definite graph Laplacian
with eigenpairs $(\lambda_h, u_h)$ and $0 = \lambda_1 \le \cdots \le \lambda_n$. For variable
$j$, $\kappa_j > 0$ is a range parameter, $\alpha_j > 0$ a smoothness parameter and
$\tau_j > 0$ the innovation standard deviation, which governs the amplitude of variable $j$. The scalar
$b_{ij}$ is the cross-dependence coefficient on the directed edge $i \leftarrow j$ of a variable graph $\GV$, and
$\sigma > 0$ is the nugget standard deviation. The full parameter is
$\theta = (\kappa, \alpha, \tau, \{b_{ij}\}, \sigma)$. We write
$D_\tau = \diag(\tau_1^2,\ldots,\tau_p^2) \otimes I_n$ for the block-diagonal innovation covariance,
$D_{\Ocal}$ for the diagonal $0/1$
indicator of observed entries and $\|y\|^2_{\Ocal} = \sum_{(j,i)\in\Ocal} y_j(s_i)^2$.

\subsection{From the Mat\'ern equation to a spatial operator}\label{sec:leps}

The Mat\'ern covariance has a differential characterization.
\citet{whittle1954stationary} showed that the stationary field on $\R^d$ solving the
fractional stochastic partial differential equation
$(\kappa^2 - \Delta)^{\alpha/2} w(s) = \mathcal{W}(s)$, with $\mathcal{W}$ Gaussian white
noise, has Mat\'ern covariance with smoothness $\nu = \alpha - d/2$. The equation, unlike the
covariance function, requires of the domain only a Laplacian, and discrete solutions on a
finite domain are Gaussian Markov random fields \citep{rue2005gaussian}. On a regular
lattice, and with $\alpha/2 \in \mathbb{N}$ so that the operator is polynomial and hence
local, the solution recovers the classical conditional autoregressions
\citep{besag1974spatial,lindgren2011explicit}. On a finite-element mesh it yields the
sparse-precision representations of \citet{lindgren2011explicit}, the basis of the SPDE
approach reviewed by \citet{lindgren2022spde}. Deep Gaussian Markov random fields extend the
lattice version \citep{siden2020deep}, and graph discretizations are analysed by
\citet{sanzalonso2022spde}. We follow the same route, with $-\Delta$ replaced by a graph
Laplacian built from the observed point cloud, so that the equation, and not a covariance
function, defines the model on an arbitrary domain.

The operator is a graph Laplacian built from local neighbourhoods.
We use the density-normalized construction of \citet{coifman2006diffusion}. Let
$\rho(s_i, s_\ell)$ be a local dissimilarity between two locations. In a Euclidean subset of
$\R^d$ this is the Euclidean distance; on a manifold or a network it is a geodesic or
shortest-path distance; and on a domain with barriers it is a distance that does not cross the
boundary, so that the construction respects the geometry of $D$. Let
$a_{i\ell} = \exp\{-\rho(s_i, s_\ell)^2/\varepsilon\}$ be a Gaussian kernel over the nearest
neighbours of each location, with degrees $q_i = \sum_\ell a_{i\ell}$, and let
$\tilde A = D_q^{-1} A D_q^{-1}$ with degrees $\tilde d_i = \sum_\ell \tilde a_{i\ell}$. The
operator is
\begin{equation}\label{eq:leps}
\Le = \varepsilon^{-1}( I - \tilde D^{-1/2} \tilde A \tilde D^{-1/2} ),
\end{equation}
symmetrized and scaled so that $n^{-1}\tr (\Le) = 1$. The bandwidth $\varepsilon$ is a fixed
multiple of the squared median distance to the $K$th nearest neighbour, with $K$ the
neighbourhood size.

Three properties motivate this choice. First, when $\rho$ is the Euclidean distance and the
locations are sampled from a smooth manifold embedded in $\R^d$, $\Le$ converges spectrally to
the Laplace--Beltrami operator of that manifold as $n \to \infty$ and $\varepsilon \to 0$
\citep{coifman2006diffusion,calder2022improved}, so the induced marginal covariance
approximates a continuum Mat\'ern covariance rather than a graph artefact
\citep{sanzalonso2022spde}. This manifold limit is what the cited theory establishes, and it
should be claimed only in the Euclidean case. Under a general dissimilarity, or when
neighbourhood edges are pruned so as not to cross a barrier, $\Le$ instead approaches a
weighted Laplacian adapted to the connectivity of the domain rather than the Laplace--Beltrami
operator of a smooth manifold; this is the intended target on a non-convex domain, where the
Euclidean operator would link points across a concavity. Second, only local neighbourhoods
enter the construction, so connectivity rather than Euclidean distance determines
dependence (\S\ref{sec:horseshoe}). Third, $\Le$
depends only on the point cloud and not on $\theta$. Every spectral quantity that estimation
requires can therefore be computed once and reused throughout the fit, which \S\ref{sec:comp}
exploits repeatedly. The combinatorial Laplacian is the simpler alternative, but its spectrum
grows with the vertex degrees and it does not converge to a Laplace--Beltrami operator under
non-uniform sampling. The choice has statistical consequences as well as numerical ones: the
bounded spectrum of the normalized operator is precisely what weakens identifiability of the
nugget (Theorem~\ref{thm:nugget}).

\subsection{Operator system, generative law and observation model}\label{sec:system}

The univariate building block writes the Whittle--Mat\'ern equation on the graph. Variable
$j$ solves
\begin{equation}\label{eq:unispde}
(\kappa_j^2 I + \Le)^{\alpha_j/2}\, w_j = \tau_j\, \xi_j, \qquad \xi_j \sim N(0, I_n),
\end{equation}
so that $w_j \sim N\{0, \tau_j^2 (\kappa_j^2 I + \Le)^{-\alpha_j}\}$. The innovation
standard deviation $\tau_j$ scales the driving noise and therefore the amplitude of the
field. The fractional power is defined by the spectral calculus,
$(\kappa^2 I + \Le)^{\alpha/2} v = \sum_h (\kappa^2 + \lambda_h)^{\alpha/2} (u_h^{\T}v)u_h$,
where $(\lambda_h, u_h)$ are the eigenpairs of $\Le$ introduced in \S\ref{sec:notation},
and the continuum correspondence is $\nu = \alpha - d/2$ for the Mat\'ern smoothness $\nu$.

Cross-variable dependence enters through a system of equations of the form
\eqref{eq:unispde}, following \citet{hu2013multivariate}. The variables are organized by the
directed acyclic graph $\GV$ on $\{1,\ldots,p\}$: an edge represents a direct dependence,
and the absence of an edge is a modelling assertion of conditional independence, made exact
by Theorem~\ref{thm:ci}. Collect the cross-dependence coefficients into the matrix
$B \in \R^{p \times p}$, with $B_{ij} = b_{ij}$ when $i \leftarrow j$ is an edge of $\GV$
and $B_{ij} = 0$ otherwise. Fixing an ordering with $j < i$ for every edge makes $B$
strictly lower triangular, so each field is generated from white noise together with the
fields already drawn. Figure~\ref{fig:schematic} spells out an example with five variables:
the source variable $1$ solves \eqref{eq:unispde} alone, variable $2$ solves
\eqref{eq:unispde} with its left side augmented by $b_{21} w_1$, variable $4$ receives such
contributions from both parents, and collecting the five equations gives the block
lower-triangular operator of panel (c). The direction is natural when the variables carry a
scientific ordering, for instance a transcription factor upstream of the genes it regulates
(\S\ref{sec:application}). When no ordering is known the choice is inconsequential for the
model's independence structure: the semantics are read from the moral graph of $\GV$, in
which two variables with a common child are joined, so reversing an edge or reordering the
variables changes the parameterization but not the encoded conditional independences
(Theorem~\ref{thm:ci}). Given such an ordering, define the block lower-triangular
operator $\bL \in \R^{N \times N}$ by
\begin{equation}\label{eq:opsys}
(\bL w)_i = (\kappa_i^2 I + \Le)^{\alpha_i/2} w_i + \sum_{(i,j) \in \GV} b_{ij} w_j
\quad (i = 1,\ldots,p).
\end{equation}
The diagonal blocks are fractional graph-Mat\'ern operators. Each off-diagonal block is the
scalar matrix $b_{ij} I$ and is present only for edges of $\GV$; an absent edge is the
structural zero $b_{ij} = 0$. The process is defined by the generative law
\begin{equation}\label{eq:gen}
\bL w = z, \quad z \sim N(0, D_\tau), \quad D_\tau = \diag(\tau_1^2,\ldots,\tau_p^2)\otimes I_n,
\end{equation}
so that $w \sim N(0, \Qw^{-1})$ with $\Qw = \bL^{\T} D_\tau^{-1}\bL$. Observations are $y = w + \epsilon$
on $\Ocal$, where $\epsilon \sim N(0, \sigma^2 I_M)$ is independent of $w$. The conditional
distribution of the field is the Gauss--Markov update
\begin{equation}\label{eq:post}
w \mid y \sim N(\mu, \Qpost^{-1}), \qquad
\Qpost = \Qw + \sigma^{-2} D_{\Ocal}, \qquad \mu = \sigma^{-2}\Qpost^{-1} \tilde y,
\end{equation}
where $\tilde y$ equals $y$ on $\Ocal$ and zero elsewhere. Each variable is driven by an
innovation of variance $\tau_j^2$. The innovation variance is a genuine amplitude parameter.
It allows the marginal variance of each variable to be fitted directly, rather than being
determined by the range and smoothness alone. Just like the setting of fitting a standard
Mat\'ern covariance function, however, under in-fill sampling only the microergodic
combination $\tau_j^2 \kappa_j^{2\nu_j}$, with $\nu_j = \alpha_j - d/2$, is estimated
consistently \citep{zhang2004inconsistent}. We make this precise in
Theorem~\ref{thm:nugget}.

The construction is triangular by design, for three reasons. Triangularity yields a proper,
recursively defined joint law with no constraint on the coefficients. It preserves
conditional independence exactly, which a full operator system does not
(Theorem~\ref{thm:ci}). It also keeps the log-determinant free of the cross-dependence
coefficients (Theorem~\ref{thm:logdet}, \S\ref{sec:est}), which is what makes the variable graph learnable.
The price is a dependence on the variable ordering, revisited in \S\ref{sec:discussion}.

We call the process defined by \eqref{eq:leps}--\eqref{eq:gen} the bigraphical
Mat\'ern--Whittle process (BMW process). The name records that the process is graphical twice over. The
spatial graph generates each variable's Mat\'ern structure through the operator, and the
variable graph carries the conditional-independence structure. The graph-domain processes of
\citet{borovitskiy2021matern} and \citet{bolin2024metric} are univariate, with a single
graph as the domain. The graphical Mat\'ern of \citet{dey2022graphical} couples univariate
Mat\'ern processes through distance-based cross-covariances on a known decomposable graph.
The present process generalizes it in three directions. The spatial model is an operator on
an arbitrary graph. The variable graph is an arbitrary directed acyclic graph. The graph is
learned from the data (\S\ref{sec:graphlearn}).

Figure~\ref{fig:schematic} summarizes the construction and displays a realization: a
coupled pair whose negative cross-dependence coefficient induces positive
cross-correlation, drawn on the horseshoe domain, which previews the behaviour on
non-convex domains studied in \S\ref{sec:horseshoe}. Within each variable, $\alpha$
controls roughness, $\kappa$ the range and $\tau^2$ the amplitude.

\begin{figure}[htbp]
\centering
\includegraphics[width=\textwidth]{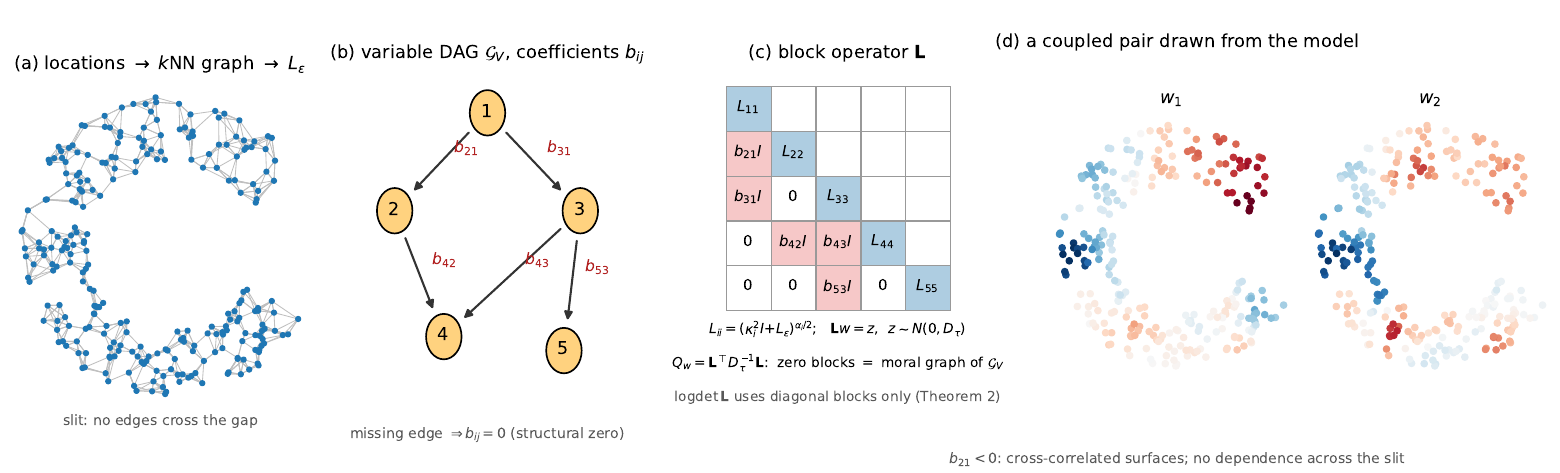}
\caption{Construction of the bigraphical Mat\'ern--Whittle process: a spatial domain
yields a graph and an operator, a variable graph supplies the cross-dependence, together
they define the block operator, and the process is drawn from it. (a) Locations on a
non-convex domain are joined into a nearest-neighbour graph that respects the domain, from
which the operator \eqref{eq:leps} is built once. (b) A directed acyclic variable graph
carries scalar cross-dependence coefficients. (c) The block lower-triangular operator of
\eqref{eq:opsys}. Zero blocks of the precision are the non-adjacencies of the moral graph,
and the log-determinant involves only the diagonal blocks. (d) A coupled pair $(w_1, w_2)$
drawn from the model on the graph of panel (a): the negative coefficient $b_{21}$ induces
positively cross-correlated surfaces, smooth along the arc and independent across the
slit.}
\label{fig:schematic}
\end{figure}

\subsection{Markov property and marginal structure}\label{sec:theory}

\begin{lemma}[Existence and validity]\label{lem:exist}
Fix locations $s_1,\ldots,s_n$ with graph Laplacian $\Le \succeq 0$. For $\kappa_i > 0$,
$\alpha_i > 0$, $\tau_i > 0$ and any directed acyclic $\GV$ with coefficients $\{b_{ij}\}$, the
operator $\bL$ of \eqref{eq:opsys} is invertible and the law \eqref{eq:gen} is the mean-zero
Gaussian $w \sim N(0, \Qw^{-1})$ with $\Qw = \bL^{\T} D_\tau^{-1} \bL \succ 0$. Because $\bL$ is
block lower-triangular, its determinant does not involve $\{b_{ij}\}$, so the law is proper for
every coefficient pattern consistent with the ordering, with no stability constraint.
\end{lemma}

An edge $i \leftarrow j$ is admitted only when $j < i$, so the coefficient block cannot
encode a directed cycle. On a Euclidean or manifold domain the same system, with $\Le$
replaced by $-\Delta$, defines a valid continuum process whenever
$\nu_i = \alpha_i - d/2 > 0$. That continuum process is what \eqref{eq:gen} discretizes. The
value of the triangular construction lies in two further exact algebraic facts. Proofs of the
results in this subsection are given in the Supplementary Material.

\begin{theorem}\label{thm:ci}
Let $\bL$ be the block lower-triangular operator \eqref{eq:opsys} with invertible diagonal
blocks, and let $\Qw = \bL^{\T} D_\tau^{-1} \bL$ with $D_\tau$ block-diagonal and positive
definite. For $i \neq j$, the $(i,j)$ block of $\Qw$ vanishes
identically if and only if neither $(i,j)$ nor $(j,i)$ is an edge of $\GV$ and $i$ and $j$
have no common child in $\GV$. The zero pattern of $\Qw$ is therefore the moral graph of
$\GV$. Whenever $i$ and $j$ are non-adjacent in the moralized variable graph,
\[
w_i \perp\!\!\!\perp w_j \mid \{w_m : m \neq i, j\},
\]
jointly across all spatial locations.
\end{theorem}

Marginal dependence between such pairs may still be non-zero, mediated by intermediate
variables. This is the intended semantics of a sparse variable graph. Triangularity is
essential. In a general operator system a zero block of $\bL$ does not produce a zero block
of $\bL^{\T} D_\tau^{-1}\bL$. Theorem~\ref{thm:ci} places the process within the graphical
Gaussian family, on an arbitrary spatial graph and with the graph encoded in $|E_V|$ scalar
parameters. It remains to identify the marginals. Write $\GV^{m}$ for the moral graph of
$\GV$ and $\mathrm{an}(i)$ for the ancestors of variable $i$.

\begin{proposition}[Bigraphical Mat\'ern--Whittle as a graphical Gaussian process]\label{prop:ggp}
Let $w$ follow \eqref{eq:gen}.
\begin{enumerate}
\item[\textup{(i)}] For every parameter value, $w$ is a graphical Gaussian process with respect to
$\GV^{m}$ in the sense of \citet{dey2022graphical}: $w_i \perp\!\!\!\perp w_j \mid \{w_l : l \neq i,j\}$
whenever $(i,j) \notin \GV^{m}$. If $\GV$ has no immoralities, in particular a chain, then $\GV^{m}$
is the skeleton of $\GV$ and $w$ is a graphical Gaussian process with respect to $\GV$ itself.
\item[\textup{(ii)}] In the eigenbasis $\Le = \sum_h \lambda_h u_h u_h^{\T}$ the process decouples
across frequencies: at $\lambda_h$ the coefficient vector $(u_h^{\T} w_1, \ldots, u_h^{\T} w_p)^{\T}$
is $N(0, \Sigma_h)$ with $\Sigma_h = L_h^{-1} D_{\tau,0} L_h^{-\T}$, $L_h = \diag\{(\kappa_i^2 +
\lambda_h)^{\alpha_i/2}\} + B$ and $D_{\tau,0} = \diag(\tau_1^2, \ldots, \tau_p^2)$. A variable with
no parents is exactly graph-Whittle--Mat\'ern, $w_i \sim N\{0, \tau_i^2(\kappa_i^2 I + \Le)^{-\alpha_i}\}$;
in general the marginal spectral density of $w_i$ is the finite mixture
$S_i(\lambda) = \sum_{j \in \{i\} \cup \mathrm{an}(i)} \tau_j^2\, g_{ij}(\lambda)^2$, where
$g_{ij}(\lambda) = [L(\lambda)^{-1}]_{ij}$ is a sum over the directed paths from $j$ to $i$. As
$\lambda \to \infty$, $S_i(\lambda) = \tau_i^2(\kappa_i^2 + \lambda)^{-\alpha_i}\{1 + o(1)\}$, so
every variable retains the Mat\'ern smoothness $\nu_i = \alpha_i - d/2$ of its own operator.
\end{enumerate}
\end{proposition}

\begin{remark}
The stitched graphical Mat\'ern of \citet{dey2022graphical} and the present process are
exact in complementary senses. Stitching preserves every univariate Mat\'ern marginal
exactly and matches cross-covariances on a finite reference set. The bigraphical
Mat\'ern--Whittle process makes the conditional independence an exact operator identity and
delivers the cross-dependence in closed form through Proposition~\ref{prop:ggp}(ii). The
price is that only source variables are exactly Mat\'ern. Descendant marginals are ancestral
spectral mixtures that retain their own smoothness.
\end{remark}

\section{Estimation}\label{sec:est}

\subsection{Exact marginal likelihood}\label{sec:lik}

\begin{figure}[t]
\centering
\includegraphics[width=\textwidth]{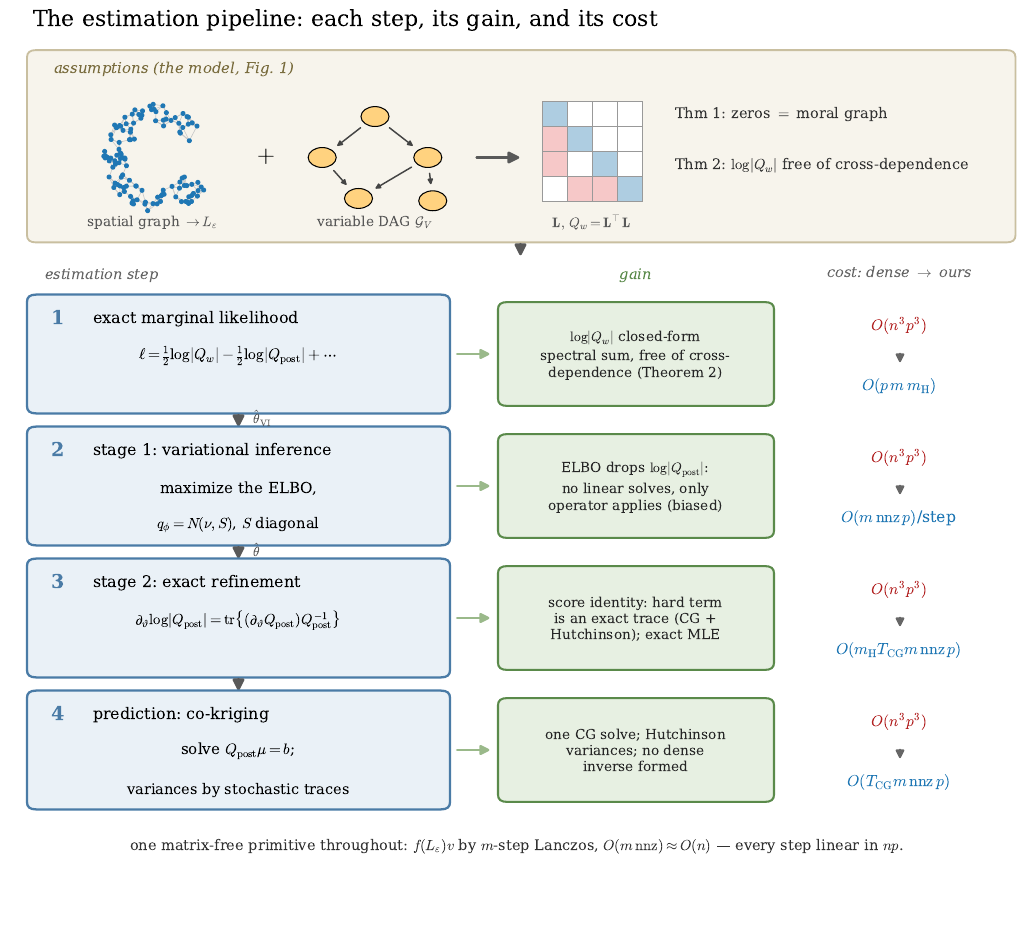}
\caption{The estimation pipeline. Each step lists its operation, the identity that makes it
cheap or exact, and its cost. The variational stage drops $\log|Q_{\mathrm{post}}|$ and all
linear solves and is fast but biased. The exact refinement turns
$\partial_\vartheta\log|Q_{\mathrm{post}}|$ into a stochastic trace and removes the bias.
Every step reduces to one matrix-free primitive and is linear in $np$, in place of the
$O(n^3p^3)$ of a dense factorization. Notation: $m$, Krylov dimension (\S\ref{sec:comp});
$\mathrm{nnz}$, non-zeros of $\Le$, of order $n$; $m_{\mathrm{H}}$, number of trace probes;
$T_{\mathrm{CG}}$, conjugate-gradient iterations (\S\ref{sec:refine}). The complexity of
every step, in $n$ and in $p$ separately, is tabulated in Section~\ref{supp:scaling} of the
Supplementary Material.}
\label{fig:method}
\end{figure}

The estimation and prediction pipeline is summarized in Fig.~\ref{fig:method}. Integrating $w$
out of \eqref{eq:gen}--\eqref{eq:post} gives the marginal log-likelihood
\begin{equation}\label{eq:loglik}
\ell(\theta) = \tfrac{1}{2}\log|\Qw| - \tfrac{1}{2}\log|\Qpost|
- M \log \sigma - \tfrac{M}{2}\log(2\pi)
- \tfrac{1}{2}\sigma^{-2} \|y\|_{\Ocal}^2
+ \tfrac{1}{2} b^{\T} \Qpost^{-1} b,
\end{equation}
with $b = \sigma^{-2}\tilde y$. We examine the terms of \eqref{eq:loglik} in turn, since the
whole estimation strategy follows from their separate difficulties. The data term
$\sigma^{-2}\|y\|_{\Ocal}^2$ and the normalizing constants are immediate. The quadratic form
$b^{\T}\Qpost^{-1} b$ requires one linear solve with $\Qpost$, which conjugate gradients
delivers from matrix--vector products alone (\S\ref{sec:refine}). The prior determinant
$\log|\Qw|$ appears to be the first obstruction, since $\Qw$ is an $N \times N$ operator with
fractional blocks. The triangular construction removes it exactly.

\begin{theorem}\label{thm:logdet}
For the operator \eqref{eq:opsys},
\[
\log |\!\det \bL| = \sum_{j=1}^{p} \frac{\alpha_j}{2}
\sum_{h=1}^{n} \log(\kappa_j^2 + \lambda_h),
\qquad
\log |\Qw| = 2 \log|\!\det \bL| - 2n \sum_{j=1}^{p} \log \tau_j .
\]
The expression does not involve the cross-dependence coefficients, and it depends on $\theta$ only through
$(\kappa,\alpha,\tau)$; the innovation variances enter solely through the separable, closed-form
term $-2n\sum_j \log\tau_j$, evaluated against the fixed spectrum of $\Le$.
\end{theorem}

Theorem~\ref{thm:logdet} has two consequences. Computationally, once the spectrum of $\Le$, or
a stochastic quadrature approximation to it, is available, the prior log-determinant and its
derivatives with respect to $\kappa$, $\alpha$ and $\tau$ are closed-form sums (\S\ref{sec:comp}).
Statistically, the cross-dependence coefficients are free: adding, deleting or perturbing edges of $\GV$ leaves
$\log|\Qw|$ unchanged, so likelihood ratios across graphs require no determinant computation.
This fact is what makes Bayesian graph learning practical in \S\ref{sec:graphlearn}.

The remaining term $\log|\Qpost|$ has no closed form. The observation indicator $D_{\Ocal}$
breaks the spectral alignment between $\Qpost$ and $\Le$, so no analogue of
Theorem~\ref{thm:logdet} exists for it, and it is the sole obstruction to exact evaluation at
scale. The estimation strategy is
organized around avoiding this term, and then confronting it only sparingly.

\subsection{First route: variational estimation, faster alternative}\label{sec:vi}\label{sec:comp}

Following the variational strategy of \citet{siden2020deep}, we introduce a mean-field
Gaussian $q_{\phi}(w) = N(\nu, S)$ with diagonal $S$. The evidence lower bound decomposes as
\[
\mathcal{L}(\theta,\phi) = E_{q}\{\log p(y \mid w, \theta)\}
+ E_{q}\{\log p(w \mid \theta)\} - E_{q}\{\log q_{\phi}(w)\} .
\]
The first expectation is $-M\log\sigma - \tfrac{1}{2}\sigma^{-2}E_q\|y-w\|_{\Ocal}^2$ up to a
constant. The second is $\tfrac{1}{2}\log|\Qw| - \tfrac{1}{2}E_q\|D_\tau^{-1/2}\bL_\theta
w\|^2$, and Theorem~\ref{thm:logdet} supplies $\tfrac{1}{2}\log|\Qw| =
\log|\!\det\bL_\theta| - n\sum_j \log\tau_j$ in closed form. The third is the Gaussian
entropy $\tfrac{1}{2}\log|S|$ plus a constant. Assembling the three parts,
\begin{equation}\label{eq:elbo}
\mathcal{L}(\theta,\phi) = \tfrac{1}{2}\log|S| - M\log\sigma + \log|\!\det \bL_{\theta}|
- n\sum_{j=1}^p \log\tau_j
- E_{q}\{ \tfrac{1}{2}\|D_\tau^{-1/2}\bL_{\theta} w\|^2
+ \tfrac{1}{2}\sigma^{-2}\|y - w\|_{\Ocal}^2 \}.
\end{equation}
The bound contains no $\log|\Qpost|$ and no linear solves. The quadratic term is a forward
operator application followed by a per-variable scaling. The expectations are handled by
reparameterization, and all parameters are optimized jointly by stochastic gradient ascent.

Three computational devices make each step cheap. First, every operation reduces to one
primitive, the application of a function of $\Le$ to a block of $p$ vectors. Running $m$
steps of the Lanczos process on $\Le$ from a vector $v$ yields an orthonormal basis $Q_m$ and
an $m \times m$ tridiagonal matrix $T_m$, and the function is applied through the small
matrix,
\begin{equation}\label{eq:krylovfun}
f(\Le)\, v \approx \|v\|\, Q_m\, f(T_m)\, e_1 .
\end{equation}
The approximation is exact for polynomials of degree $m-1$ and its error decays geometrically
in $m$ for analytic $f$ (Proposition~\ref{prop:exact} and Theorem~\ref{thm:kryerr} of the
Supplementary Material). The Krylov subspace is invariant to the shifts $\kappa_j^2 I$
(Lemma~\ref{lem:shift}), so one Lanczos run per vector serves every variable and every
required function. The $p$ variables are processed in lock-step, so one application costs
$O(m\,\mathrm{nnz}\,p)$, where $\mathrm{nnz}$ is the number of non-zeros of $\Le$, of order
$n$. No dense matrix, eigendecomposition or Cholesky factor is ever formed, and memory stays
$O(np)$.

Second, the prior determinant is a one-time computation. Theorem~\ref{thm:logdet} reduces
$\log|\Qw|$ and its derivatives to sums of scalar functions over the spectrum of $\Le$. The
Lanczos recurrence supplies these sums as a Gauss quadrature rule with the Ritz values as
nodes (Theorem~\ref{thm:gauss}), averaged over Rademacher probes as in stochastic Lanczos
quadrature \citep{ubaru2017fast,chen2021analysis}. The nodes and weights do not depend on
$\theta$. They are computed once, before optimization begins, and every later evaluation of
$\log|\Qw|$ and its derivatives is a closed-form sum over fixed nodes. The setup took
$21$ seconds at $n = 1{,}200{,}000$.

Third, gradients flow through the fractional operator by a custom differentiation rule. The
forward pass applies \eqref{eq:krylovfun}. The backward pass returns the input gradient by
self-adjointness and the gradients in $\kappa_i$ and $\alpha_i$ by one further Krylov
application each, as in the matrix-free differentiation of \citet{gardner2018gpytorch}. One
variational step therefore costs a constant number of operator applications. Measured step
times are reported in \S\ref{sec:comptiming} and Section~\ref{supp:scaling} of the
Supplementary Material.

Mean-field variational inference is fast but biased. The diagonal $S$ understates the
conditional covariance, which compresses the spatial parameters and the innovation variances.
The refinement below removes this bias for $(\kappa,\alpha,\tau,b)$. The nugget requires
separate discussion (\S\ref{sec:nugget}).

\subsection{Second route: exact likelihood refinement}\label{sec:refine}\label{sec:nugget}

Starting from the variational estimate, we maximize the exact likelihood \eqref{eq:loglik}
by gradient ascent. The
score identity
$\nabla_\theta \ell = E_{w \mid y, \theta}\{\nabla_\theta \log p(w, y \mid \theta)\}$ gives,
for a parameter $\vartheta$ entering $\Qw$,
\begin{equation}\label{eq:grad}
\partial_{\vartheta} \ell = \tfrac{1}{2} \partial_{\vartheta} \log|\Qw|
- \tfrac{1}{2} [ \mu^{\T} (\partial_{\vartheta}\Qw) \mu
+ \tr\{ (\partial_{\vartheta}\Qw) \Qpost^{-1} \} ].
\end{equation}
The obstruction $\log|\Qpost|$ is dissolved rather than met. Its derivative is exactly the
trace in \eqref{eq:grad}, because $\Qpost = \Qw + \sigma^{-2} D_{\Ocal}$ depends on
$\vartheta$ only through $\Qw$, so the refinement ascends the exact gradient using only
linear solves. The component derivatives of $\bL$ and of $\log|\Qw|$ in $\kappa_i$,
$\alpha_i$, $b_{ij}$, $\tau_i$ and $\sigma$ are closed-form and matrix-free. Their
expressions are collected in Section~\ref{supp:krylov} of the Supplementary Material.

Each gradient step requires one conjugate-gradient solve for the posterior mean and
$m_{\mathrm{H}}$ stochastic trace solves with Rademacher probes
\citep{hutchinson1990stochastic}. The trace solves are computed once per step and reused
across all parameters. Conjugate gradients uses only the product
$\Qpost w = \bL^{\T} D_\tau^{-1}(\bL w) + \sigma^{-2} D_{\Ocal} w$. In the graph regime the
spectrum of $\Le$ is bounded uniformly in $n$, so $30$ to $40$ iterations sufficed, and
standard preconditioning restores an $n$-independent count in fine-mesh regimes
\citep{cutajar2016preconditioning}. Predictive variances use the same trace machinery, and
posterior draws use the perturbation sampler of \citet{papandreou2010gaussian}. The
refinement was validated against a sparse-Cholesky reference and against finite differences
of \eqref{eq:loglik} at small $n$.

A refinement step is roughly an order of magnitude more expensive than a variational step,
at $O(m_{\mathrm{H}} T_{\mathrm{CG}}\, m\, \mathrm{nnz}\, p)$ against
$O(m\, \mathrm{nnz}\, p)$, where $T_{\mathrm{CG}}$ is the conjugate-gradient iteration
count. Both are linear in $n$ and in $p$. A step-by-step complexity table and measured
timings are given in \S\ref{sec:comptiming} and Section~\ref{supp:scaling} of the
Supplementary Material.

The nugget requires a separate analysis. The variance $\sigma^2$ of the measurement error
governs predictive uncertainty, and the operator class limits the information the data
carry about it. Under full observation the projection $z_h = u_h^{\T} y$ on the
eigenvectors of $\Le$ decouples the data into $n$ independent spectral coefficients
$z_h \sim N\{0, s(\lambda_h) + \sigma^2\}$ with signal density
$s(\lambda) = \tau^2(\kappa^2 + \lambda)^{-\alpha}$, and the nugget enters the
likelihood only through modes whose signal has decayed to or below $\sigma^2$. Proofs of
the following results are in the Supplementary Material.

\begin{theorem}\label{thm:nugget}
Consider the univariate model $y = w + \epsilon$ with
$w \sim N\{0, \tau^2(\kappa^2 I + \Le)^{-\alpha}\}$ and $\epsilon \sim N(0, \sigma^2 I)$
independent, where $s(\lambda) = \tau^2(\kappa^2+\lambda)^{-\alpha}$. The Fisher information for $\sigma^2$ is
\[
\mathcal{I}(\sigma^2) = \frac{\neff}{2\sigma^4}, \qquad
\neff = \sum_{h=1}^{n} \pi_h^2, \qquad
\pi_h = \frac{\sigma^2}{s(\lambda_h) + \sigma^2} \in (0,1],
\]
and any unbiased estimator satisfies
$\{\mathrm{var}(\hat\sigma^2)\}^{1/2} \ge \sigma^2 (2/\neff)^{1/2}$.
\end{theorem}

\begin{corollary}\label{cor:ceiling}
Suppose $\lambda_{\max} = \sup_n \lambda_n < \infty$ and define
$s_{\min} = \tau^2(\kappa^2 + \lambda_{\max})^{-\alpha}$. Then
\[
\frac{\neff}{n} \le \Bigl( \frac{\sigma^2}{s_{\min} + \sigma^2} \Bigr)^{2},
\]
a constant that does not tend to one under in-fill sampling. In particular, if
$\sigma^2 \ll s_{\min}$ then $\neff/n \lesssim (\sigma^2/s_{\min})^2$. Consistency of
$\hat\sigma^2$ is retained, since $\neff \to \infty$, but the relative efficiency is pinned
below a fixed ceiling, so further in-fill data do not improve the relative precision.
\end{corollary}

\begin{remark}\label{rem:nugget}
Three consequences follow. First, a parallel phenomenon governs the innovation variance:
under in-fill sampling the amplitude and the range are not separately identifiable, only the
smoothness and the microergodic combination $\tau^2 \kappa^{2\nu}$, with
$\nu = \alpha - d/2$, the graph analogue of classical in-fill results for the Mat\'ern
\citep{stein1999interpolation,zhang2004inconsistent}. Section~\ref{sec:sims} therefore
reports recovery of the identified quantities, the amplitude, the implied marginal variance
and the microergodic combination, with the individual range and smoothness in the
Supplementary Material. Second, Theorem~\ref{thm:nugget} extends
\citet{tang2021identifiability} to the free-smoothness, bounded-spectrum regime of graph
operators, and exposes a trade-off inside the model class: the normalization that makes the
marginal a faithful continuum Mat\'ern covariance (\S\ref{sec:leps}) is precisely what
bounds the spectrum, truncating the high-frequency tail through which the nugget is seen.
Third, the estimator of \S\ref{sec:vi} inherits this geometry, since the bound
\eqref{eq:elbo} exerts almost no pressure on $\sigma$: in a replicated sweep the variational
estimate sat at its optimization floor while the truth varied over a $7.5$-fold range
(Table~\ref{tab:signr} of the Supplementary Material). We therefore report $\sigma$ from
the first stage and, when the nugget is of scientific interest and the estimated $\neff$
exceeds a small threshold, empirically $\neff \ge 5$, rerun the refinement above with
$\sigma$ free. Held-out prediction error changed by less than $3\%$ across a $16$-fold
range of $\sigma$, so the treatment of the nugget has little effect on prediction.
\end{remark}

\section{Learning the variable graph}\label{sec:graphlearn}\label{sec:eigenlik}\label{sec:sampler}\label{sec:mcmcconv}

A natural objection to any graphical spatial model is that the variable graph must be known
in advance. Here it need not be, because scoring a candidate graph is cheap. By
Theorem~\ref{thm:logdet} the prior determinant does not depend on the cross-dependence
coefficients, so competing graphs differ only in the quadratic part of the likelihood. The
likelihood also factorizes over spatial frequencies. Let $\Le u_h = \lambda_h u_h$ be the
eigenpairs of the common operator. Projecting the observed fields onto this basis gives the
spectral data
\begin{equation}\label{eq:speccoef}
z_h = (u_h^{\T} y_1, \ldots, u_h^{\T} y_p)^{\T} \in \R^p, \qquad h = 1, \ldots, n .
\end{equation}
At frequency $\lambda_h$ the operator system reduces to the lower-triangular
$p \times p$ matrix $L_h = D_h + B$, where
$D_h = \diag\{(\kappa_i^2+\lambda_h)^{\alpha_i/2}\}$ carries the spatial structure and $B$ is
the coefficient matrix of \S\ref{sec:system}. With $D_{\tau,0} = \diag(\tau_1^2,\ldots,\tau_p^2)$, the modes
are independent $p$-variate Gaussians and, under full observation,
\begin{equation}\label{eq:eigenlik}
\ell(\theta) = -\tfrac{1}{2}\sum_{h=1}^{n}
\{ \log |\Sigma_h| + z_h^{\T} \Sigma_h^{-1} z_h \} + \mathrm{const}, \qquad
\Sigma_h = (L_h^{\T} D_{\tau,0}^{-1} L_h)^{-1} + \sigma^2 I_p .
\end{equation}
A candidate graph is therefore scored by $n$ independent $p$-variate densities rather than by
one $np$-dimensional density, and changing an edge changes only $B$. Direct evaluation costs
$O(np^3)$, which is milliseconds for $p$ up to about $40$. The cubic factor comes only from
the nugget, which densifies each $\Sigma_h$. Writing
$\log|\Sigma_h| = \log|\sigma^2 Q_h + I| - \log|Q_h|$ with
$Q_h = L_h^{\T} D_{\tau,0}^{-1} L_h$ removes it, since the sparse Cholesky factor of $Q_h$
inherits the moral graph and shares one symbolic factorization across all $n$ modes. The
sparse form agreed with the dense one to $10^{-8}$ and was about $230$ times faster at
$p = 1{,}600$.

Each candidate edge $e$ carries an inclusion indicator $\gamma_e$ and a coefficient $b_e$
with a spike-and-slab prior: $b_e = 0$ when $\gamma_e = 0$, $b_e \sim N(0, \tau^2)$ when
$\gamma_e = 1$, and $\gamma_e \sim \mathrm{Ber}(\pi)$
\citep{mitchell1988bayesian,george1993variable}. The sampler draws the indicators
$\gamma_e$ and the coefficients $b_e$ jointly, and inference targets the marginal
edge-inclusion probabilities. The per-variable spatial parameters and the nugget,
$(\kappa, \alpha, \tau, \sigma)$, are not sampled: they are fixed at estimates from a
preliminary joint fit of \S\ref{sec:est} under the candidate set. The plug-in does not
require the true graph. By Proposition~\ref{prop:ggp}(ii) every ancestral term of the
marginal spectrum decays strictly faster than the variable's own term, so the spectral tail
that identifies the range, smoothness and amplitude is unaffected by the edge set, and the
plugged-in amplitudes agree with the residual variances recovered inside the selection
(\S\ref{sec:application}). Sampling the spatial parameters and the nugget jointly with the graph is
deferred to future work (\S\ref{sec:discussion}).

At large $p$ the candidate set is first thinned by a partial-correlation screen: partial
correlations are computed from a ridge-regularized inverse of the empirical correlation
matrix over the locations, and each variable retains its few strongest earlier-ordered
partners above a fixed threshold. The screen is tuned for recall, so it retains a superset
of the plausible edges and the sampler prunes the excess; the constants are given in the
Supplementary Material.

The number of active edges varies, so the sampler is a reversible-jump chain with birth and
death moves \citep{green1995reversible}. The posterior over graphs is multimodal, because
distinct edge sets can explain the same field almost equally well, and two devices counter
this. An edge-swap move exchanges one active edge for one inactive edge; the edge count is
preserved, the prior and proposal terms cancel, and the acceptance probability is the
tempered likelihood ratio alone. Parallel tempering
\citep{geyer1991markov,earl2005parallel} lets hotter replicas escape modes and pass
improved graphs to the colder replicas, and the agreement of two independent tempered runs
is itself a convergence check. Every likelihood evaluation is \eqref{eq:eigenlik}, and by
Theorem~\ref{thm:logdet} a proposed edge leaves the prior determinant unchanged. The full
sampler, its move mixture, the acceptance ratio of every move and all settings are given in
Section~\ref{supp:sampler} of the Supplementary Material (Algorithm~\ref{alg:pt}).

The sampler fixes the variable ordering and selects edges within it. The ordering itself
could in principle be learned: the eigenmode likelihood \eqref{eq:eigenlik} is well defined
for any directed acyclic graph under any ordering, since the determinant of $L_h$ is the
product of its diagonal entries for every acyclic coefficient pattern, and
Theorem~\ref{thm:logdet} holds unchanged, so an order move costs no more per evaluation
than an edge move. The obstacle is the search space, which grows by the factor $p!$, with a
correspondingly heavier burden on mixing. We therefore fix the ordering and return to order
learning in \S\ref{sec:discussion}.

\section{Simulation studies}\label{sec:sims}

\subsection{Design}\label{sec:simdesign}

Table~\ref{tab:design} defines the settings, which vary the specification, the strength of
the cross-variable dependence, and the dimensions. In every setting the operator is the normalized Laplacian of \S\ref{sec:leps} on
a $15$-nearest-neighbour graph, with $\kappa_j \sim \mathrm{Un}(0.25, 0.45)$,
$\alpha_j \sim \mathrm{Un}(1.6, 2.8)$, heteroscedastic amplitudes
$\tau_j \sim \mathrm{Un}(0.5, 2.0)$, and signal-to-noise ratio $10$. The cross-dependence
coefficients are $b_{j,j-1} = -\rho_j\,\kappa_j^{\alpha_j}$ along a chain, with $\rho_j$ drawn
from the interval in the table. Each replicate holds data out twice. The primary hold-out removes forty-five percent of the
variable--location pairs at random, so prediction can interpolate. The structured hold-out
removes variables over a contiguous region, and the model is refitted from the reduced data,
so no fit ever sees its test set. Five estimators run in every setting subject to
documented feasibility caps: (i) the two-stage estimator of \S\ref{sec:est}, denoted BMW-ML;
(ii) its first stage alone, BMW-VI; (iii) the same fit with the variable graph removed,
BMW-Ind; (iv) the graphical Mat\'ern of \citet{dey2022graphical} with per-variable sills
estimated, denoted GM, the term graphical Gaussian process being reserved for the general
class; and (v) an oracle, which
predicts with the generating parameters and calibrates every comparison. There are
fifty replicates per setting, five at $n = 10^4$. All runs are seeded and reproducible.
Four further settings, G1--G4, assess Bayesian graph learning and are described in the
Supplementary Material. The
presentation features three settings, L1, L3 and the non-convex domain H. The remaining
settings of Table~\ref{tab:design} are summarized briefly, and their complete results are in
Section~\ref{supp:sims} of the Supplementary Material.

\begin{table}[htbp]
\tbl{Simulation settings for the demonstrations of \S\ref{sec:sims}. Strong, moderate and
weak denote $\rho_j$ drawn from $(0.80, 0.98)$, $(0.45, 0.65)$ and $(0.15, 0.35)$; mixed
draws each $\rho_j$ across the full range. A1--A3 fix $np = 6\times10^5$ and trade
locations for variables}{%
\begin{tabular}{llllll}
Setting & $n$ & $p$ & Truth & Cross-dependence & Structured hold-out \\[3pt]
\multicolumn{6}{@{}l}{\emph{Key settings: recovery of the identified parameters and cross-variable borrowing}}\\[2pt]
S1--S3 & 500 & 8 & correct & strong/moderate/weak & half-domain \\
Smix & 500 & 8 & correct & mixed & half-domain \\
S4, S4mix & 800 & 20 & correct & strong, mixed & half-domain \\
L1 & 1{,}200 & 60 & correct & strong/moderate/mixed & band \\
L2 & 2{,}000 & 200 & correct & strong/moderate/mixed & band \\
L3 & 10{,}000 & 100 & correct & strong/moderate/mixed & band \\
H & 800 & 5 & correct, horseshoe domain & strong & tips \\[2pt]
\multicolumn{6}{@{}l}{\emph{Scaling and stress: many variables and extreme aspect ratios}}\\[2pt]
P1--P5 & 2{,}000 & 50--800 & correct & strong & band \\
A1--A3 & 2{,}000--300 & 300--2{,}000 & correct & strong & band \\[2pt]
\multicolumn{6}{@{}l}{\emph{Robustness: data generated outside the model class}}\\[2pt]
M1 & 500 & 5 & Euclidean operator & strong & band \\
M1b & 500 & 5 & Euclidean, distorted cross-covariance & strong & band \\
M2 & 300 & 20 & block coregionalization & --- & band \\
\end{tabular}}
\label{tab:design}
\end{table}

\subsection{Recovery and cross-variable borrowing}\label{sec:recovery}\label{sec:predsim}\label{sec:misspec}

Under in-fill asymptotics the pairs $(\kappa_j, \alpha_j)$ and $(\tau_j, \kappa_j)$ are only
weakly identified, as in all Mat\'ern-type models
\citep{zhang2004inconsistent,tang2021identifiability}, and \S\ref{sec:nugget} identifies the
functions that the data determine: the amplitude $\tau_j$, the marginal variance $V_j$, the
microergodic parameter $\theta_j = V_j \kappa_j^{2\alpha_j - d}$, and the cross-dependence
coefficients $b_{jk}$. We report recovery of these quantities. Individual $\kappa_j$
and $\alpha_j$ estimates ride the likelihood ridge, and their correlations are given in
the Supplementary Material.

At L1, with $n = 1{,}200$ and $p = 60$, the median recovery correlations over fifty
replicates are $0.98$ for $\tau$, $0.97$ for $V$ and $\theta$, and $0.86$ for the
cross-dependence coefficients (Fig.~\ref{fig:identified}). At the grand scale L3, with
$n = 10^4$, $p = 100$ and five replicates, the correlations are $0.97$ to $0.99$, with
held-out error within one percent of the oracle. Interval coverage on the primary hold-out
is $0.89$ to $0.91$ against the nominal $0.90$ in every setting, matching the oracle
(Table~\ref{stab:recov} of the Supplementary Material). The variable graph contributes most
where information must cross variables. On the structured hold-out at L1, where every
eighth variable vanishes over a wide band, BMW-ML reduces mean squared error by
$47\%$ relative to its graph-free counterpart. On the primary misaligned hold-out the
gains are under $8\%$, because nearby observations of the same variable suffice for
interpolation. The contrast between the two designs, not either number alone, is the
finding. The remaining settings of Table~\ref{tab:design} behave consistently, and we record only
their headlines; complete boxplots, prediction tables and coverage for every setting are in
Fig.~\ref{sfig:identified8}, Fig.~\ref{sfig:cokrig} and
Tables~\ref{stab:recov}--\ref{stab:pred} of the Supplementary Material. Across the eight
fifty-replicate settings with strong cross-dependence, the median correlations exceed
$0.95$ for $\tau$ and $0.86$ for $V$ and $\theta$; Figure~\ref{fig:identified} includes the
smallest setting S1, the largest dimension P5 and the stress shape A3, which is weakest
because each of its $2{,}000$ variables is informed by only $300$ spectral coefficients of
the form \eqref{eq:speccoef}. Structured-hold-out gains reach $82\%$ at S4, $50\%$ at P5
and $89\%$ under the block-coregionalization truth M2, while the fit stays within one
percent of the oracle in every well-specified setting. Weak cross-dependence gives the
expected null. At S3 the gain is $-1.0\%$ with Monte Carlo standard error $2.8$, since a
weak coefficient transfers little information across variables, and the fit remains within
one percent of the oracle, so estimating the graph incurs no measurable cost. Misspecification is
asymmetric. Under the Euclidean-operator truth M1, which favours the graphical Mat\'ern,
that model predicts best, as it should, and the proposed fit concedes twenty-eight percent
in mean squared error while covering at $0.93$ against the graphical Mat\'ern's $0.67$.
Under the block coregionalization M2 the graphical Mat\'ern attains root mean squared
error $1.60$ with coverage $0.33$, against $0.346$ for the proposed fit and $0.326$ for
the oracle.

\begin{figure}[htbp]
\centering
\includegraphics[width=\textwidth]{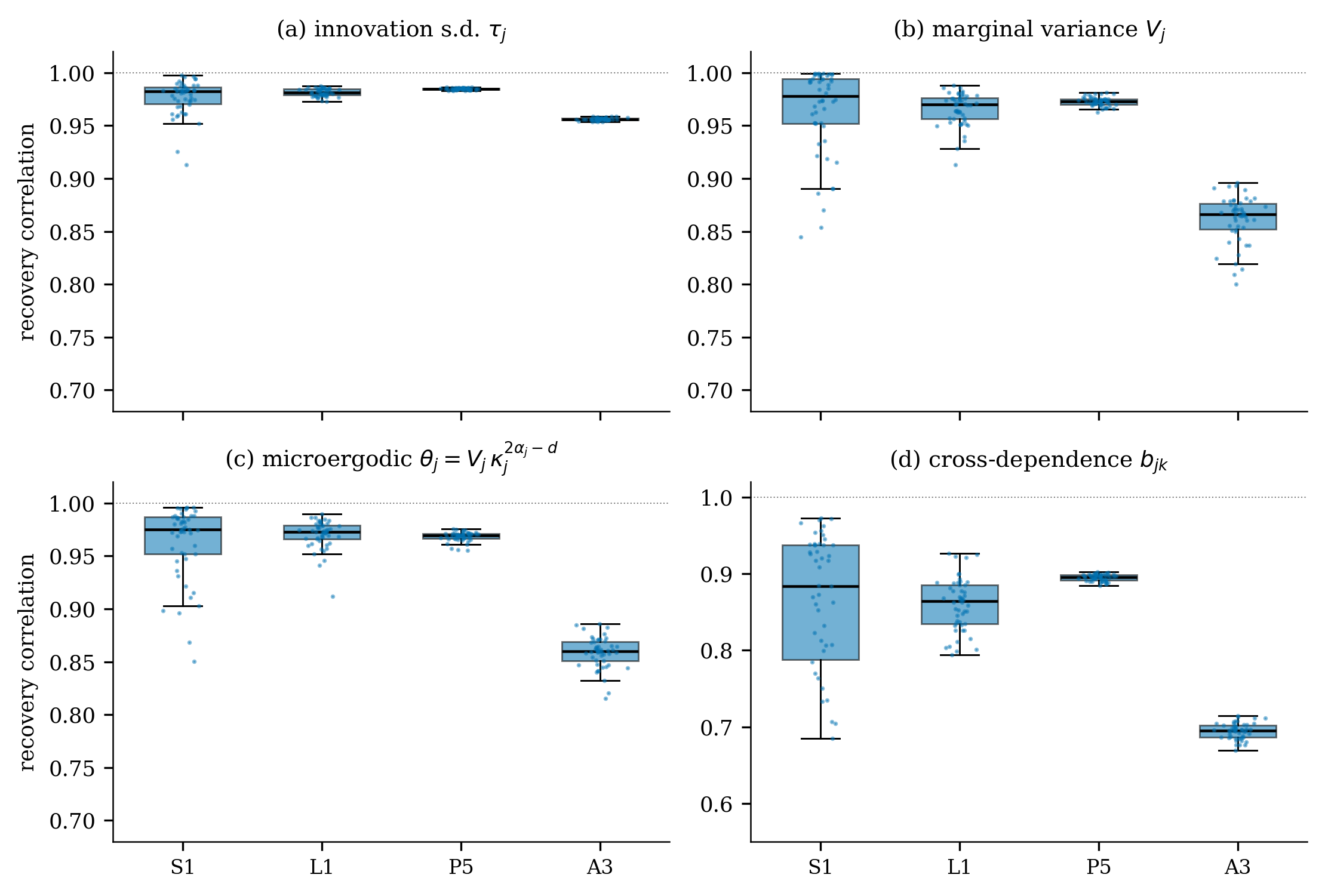}
\caption{Recovery of the identified parameters by BMW-ML at four
representative settings. Each box holds fifty per-replicate correlations between estimate
and truth. Panels show $\tau_j$, $V_j$, $\theta_j$ and $b_{jk}$. All four settings share
strong cross-dependence, so scale is the only quantity varying: the smallest setting S1
($n = 500$, $p = 8$), the featured L1, the largest dimension P5 ($p = 800$) and the stress
shape A3 ($p = 2{,}000$ at $n = 300$). The individual $\kappa_j$ and $\alpha_j$ are weakly
identified and are not shown. All eight fifty-replicate settings appear in
Fig.~\ref{sfig:identified8} of the Supplementary Material.}
\label{fig:identified}
\end{figure}

\subsection{Graph recovery}\label{sec:graphsim}

Four settings assess the Bayesian edge selection of \S\ref{sec:sampler} against known
truth. Any retention threshold on the inclusion probabilities is arbitrary, so we summarize
each setting by the separation of the posterior: the mean inclusion probability of true
edges against that of absent candidate pairs, with the area under the edge-ranking curve as
a threshold-free companion. The separation is wide throughout. Under a
block-coregionalization truth the mean probability is $0.75$ for within-block pairs against
$0.06$ for the $150$ truly independent pairs, with ranking area $0.999$; under a
misspecified Euclidean operator it is $0.64$ against $0.13$. Against increasing edge
density at $p = 40$, and through $p = 320$ with the partial-correlation screen, true edges
average $0.40$ to $0.66$ against at most $0.14$ for absent pairs, and the ranking area
stays between $0.89$ and $0.95$, the dips tracking the screen's recall rather than the
sampler. Per-setting separations, figures and tables are in Section~\ref{supp:graph} of
the Supplementary Material.

Parallel tempering is what secures Monte Carlo convergence at high edge density. At $2p$
active edges, independent single-edge-flip chains agreed poorly per edge, with a
between-chain correlation of the inclusion probabilities of $0.13$; the tempered sampler
with the swap move reduced the number of disagreeing edges from $112$ to $33$ in the
hardest case. The edge ranking was reproducible across runs throughout, and false-edge
probabilities remained pinned near the prior. We report per-edge posteriors where two
independent tempered runs agree, and the ranking otherwise.

\subsection{Computational performance}\label{sec:comptiming}

The central computational claim of the paper is per-step cost linear in both $n$ and $p$,
and the simulations measure it directly. All timings in this subsection were measured on a
single machine, an Apple M4 laptop with ten cores, four performance and six efficiency, and
$16$ gigabytes of unified memory. The variational step took
$27$ milliseconds at $np = 4 \times 10^{4}$ and $156$ milliseconds at $np = 4 \times 10^{5}$,
and the one-time spectral setup stayed at or below a tenth of a second at these sizes
(Table~\ref{tab:timing} of the
Supplementary Material). The same proportionality holds in $p$ at fixed $n$:
a complete variational fit at $n = 2{,}000$ took $9$ seconds at $p = 50$ and $135$ seconds
at $p = 800$, a sixteen-fold increase in dimension for a fifteen-fold increase in time. The
exact refinement scaled in proportion, at thirty to forty times the variational cost per
fit. At the largest sizes the operation count remains linear but memory traffic sets the
constant. At $n = 1{,}200{,}000$ with $p = 3$ the spectral setup took $21$ seconds and one
step eleven seconds, for a complete variational fit in $56$ minutes, and the application of
\S\ref{sec:application} fitted $np = 2.2 \times 10^{7}$ in $75$ minutes.

The comparison with the dense competitor quantifies a difference in kind. At $n = 400$, one
likelihood evaluation of the graphical Mat\'ern took $0.19$ seconds at $p = 3$ and
$1.0$ second at $p = 40$, at which point one variational step took $12$ milliseconds, a
factor of $84$. Beyond $p = 40$ the graphical Mat\'ern was not practical, while the
variational step was measured through $p = 320$ at $77$ milliseconds
(Fig.~\ref{sfig:highp} of the Supplementary Material). Per-replicate times for every setting
of Table~\ref{tab:design} are reported in Table~\ref{stab:timing} of the Supplementary
Material.

\subsection{Non-convex domains}\label{sec:horseshoe}

Setting H of Table~\ref{tab:design} places $n = 800$ locations on the horseshoe-shaped
domain of \citet{ramsay2002spline}, a curved non-convex region in which locations on
opposite sides of the central slit are close in Euclidean distance yet far apart along the
domain. The underlying field varies smoothly along the region, and the experiment is
replicated fifty times. Euclidean distance misrepresents this geometry: across the slit its
correlation with the true along-domain separation is $0.03$,
while the graph geodesic built from the same local neighbourhoods achieves $0.98$.
Figure~\ref{fig:horseshoe} shows the consequence. A Euclidean Mat\'ern leaks correlation
across the gap. The graph-operator correlation follows the arc and vanishes across the slit.
Holding out both tips, the graph operator attains mean squared error $0.082 \pm 0.004$ against
$0.182 \pm 0.010$ for the Euclidean model, a mean reduction of $52\%$. On real data the true
metric is rarely known, and the graph operator lets connectivity define proximity. The
caveat is that the graph must respect the domain, by pruning, meshing or given network
neighbours \citep{verhoef2010moving,bakka2019non,niu2019intrinsic}.

Taken together, the simulations support one conclusion. In settings where standard
multivariate models are applicable, the proposed process attains comparable accuracy and
approaches the oracle. In settings where the domain geometry renders a Euclidean covariance
invalid, the proposed process remains accurate.

\begin{figure}[htbp]
\centering
\includegraphics[width=\textwidth]{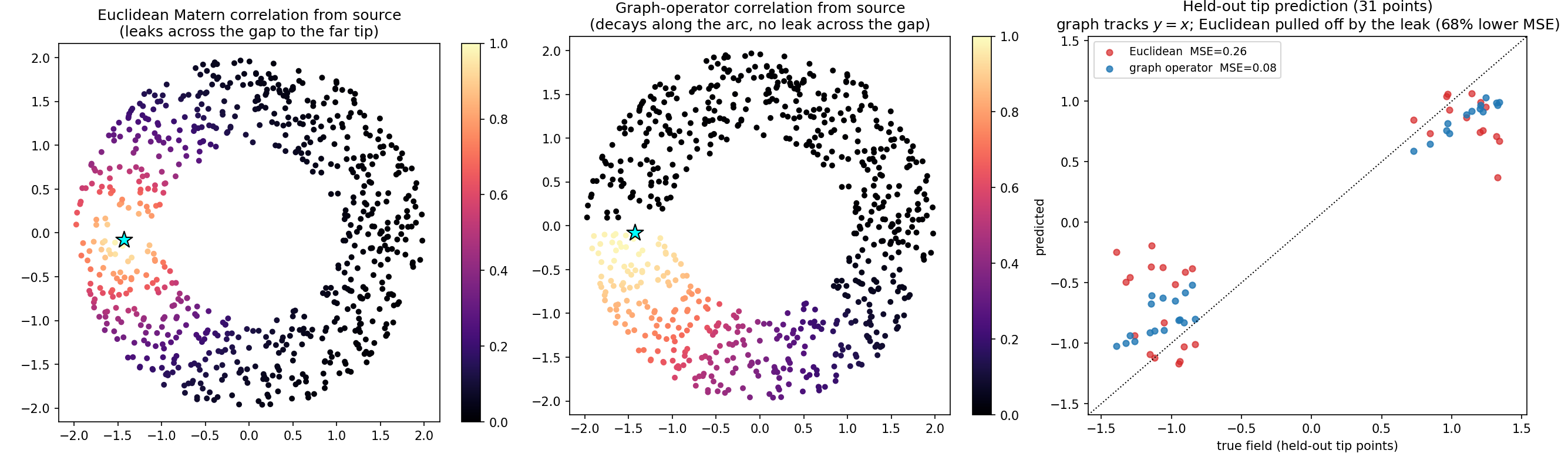}
\caption{Prediction on a non-convex domain; the star marks a source point at one tip of the
horseshoe. Left: the Euclidean-Mat\'ern
correlation from the source, which leaks across the gap. Centre: the graph-operator
correlation, which decays along the arc and vanishes across the gap. Right: held-out
tip prediction for one replicate; over fifty replicates the graph operator attains mean
squared error $0.082$ against $0.182$ for the Euclidean model.}
\label{fig:horseshoe}
\end{figure}

\section{Application: spatial transcriptomics of the mouse brain}\label{sec:application}

We apply the method to imaging-based spatial transcriptomics of the mouse brain, from the
whole-brain atlas of \citet{zhang2023molecular} distributed through the Allen Brain Cell Atlas.
A multiplexed error-robust in-situ hybridization assay measures a panel of $1{,}122$ genes at
single-cell resolution. We analyse one coronal section of a single adult female specimen, a
cerebellar and hindbrain slice with $n = 19{,}809$ cells, retaining all $p = 1{,}122$ genes, so
the latent field carries $np = 2.2 \times 10^{7}$ space--variable pairs. The
section is non-convex, with ventricles and white-matter tracts, and its dominant populations
are cerebellar granule neurons, molecular-layer interneurons, glia and oligodendrocytes.

The genes are organized by a variable graph. We fix an ordering by decreasing marginal
variance. The ordering only parameterizes the model, and the reported object is the
undirected graph of Theorem~\ref{thm:ci}, so no regulatory direction is claimed. Candidate
edges are selected by the partial-correlation screen of \S\ref{sec:graphlearn}: each gene
admits at most its three strongest earlier-ordered partners with absolute partial
correlation above $0.08$, giving $305$ candidates. Expression counts are centred and placed
on a single common scale, not standardized gene by gene, because the per-gene amplitudes
$\tau_j$ are part of the model. Each cell centroid is joined to its nearest neighbours to
form the operator \eqref{eq:leps}. The cell centroids sample the tissue densely, so
nearest-neighbour edges do not cross anatomical boundaries and the operator follows the
geometry of the section.

The estimator of \S\ref{sec:est}, run through its variational stage BMW-VI, fits the
per-gene ranges, smoothnesses and amplitudes together with the cross-dependence
coefficients and the nugget. The fit stopped by its convergence criterion after $300$
steps, in $75$ minutes on the four
performance cores of the Apple M4 laptop of \S\ref{sec:comptiming}, about $0.2$ milliseconds per space--variable pair;
graph construction, the Lanczos spectrum and the candidate screen together took
under half a second. The estimated amplitudes span an order of magnitude, $\hat\tau_j$ from
$0.64$ to $5.65$, and the largest belong to genes with sharp, spatially confined expression;
the estimated nugget is $\hat\sigma = 0.19$ on the common scale. Per-gene estimates for the full
panel are tabulated in the Supplementary Material.

We validate the fit by strictly out-of-sample co-kriging. The targets are four cell-type
marker genes chosen a priori, two oligodendrocyte markers and two neuronal markers, so that
each of the section's two dominant compartments contributes two held-out genes with strong
and distinct spatial structure. Each gene is held out over a
contiguous band covering about a third of the section; the hold-out is fixed before any
data-dependent step, the centring, the common scale and the candidate screen are computed
from training cells only, and the model is fitted without ever seeing the held-out band. The
fitted model then predicts each gene over its band from the co-expressed genes, against a
univariate baseline that sets the cross-dependence coefficients to zero. Borrowing across
genes reduced held-out mean squared error by $91\%$ for the myelin gene \textit{Mog} and
$88\%$ for \textit{Cldn11}, and by $50\%$ and $52\%$ for the neuronal markers \textit{Pvalb}
and \textit{Zic1}, a mean reduction of $70\%$. Figure~\ref{fig:merfish} shows the
reconstruction of \textit{Mog}: the multivariate prediction recovers the held-out pattern
that the univariate fit, having no observation of the gene in the band, cannot.

The final panel of Fig.~\ref{fig:merfish}(f) shows a comparison with a rank-$10$ spatial
factor model, a popular choice for modelling and prediction problems such as this one and
the statistical core of the factor methods used in spatial transcriptomics
\citep{wang2003generalized,velten2022mefisto,townes2023nonnegative}. The factor model
reduces held-out mean squared error by at most $38\%$, and by less than $8\%$ for the
myelin genes, so the multivariate dependence structure identified by the proposed model
offers dramatically improved predictive performance. Implementation details and results for
every rank are in Section~\ref{supp:data} of the Supplementary Material.

The estimates are scientifically coherent. The four held-out genes carry among the largest
amplitudes in the panel, with $\hat\tau$ between $2.8$ and $5.5$. The heaviest
cross-dependences link the myelin genes \textit{Mog} and \textit{Cldn11} to each other and
to the oligodendrocyte-lineage regulator \textit{Sox10} \citep{stolt2002sox10}. This
mirrors established biology: \textit{Mog} and \textit{Cldn11} encode structural components
of the myelin sheath and localize to white-matter tracts, while \textit{Pvalb} marks
Purkinje cells and molecular-layer interneurons and \textit{Zic1} the granule-cell lineage
of the cerebellar cortex \citep{celio1990calbindin,aruga1994novel}. The learned graph thus
separates, without supervision, the tract and cortical compartments that anatomy prescribes
for this section. The
estimated graph is a spatial analogue of a gene co-expression network
\citep{zhang2005wgcna,acharyya2022spacex}. Dependence is strongest among markers of a common
cell type, which co-localize in space. Recovery of this module without supervision is the
real-data counterpart of the structured-hold-out gains of \S\ref{sec:predsim}.

Finally, the variable graph can be inferred rather than fixed by the screen. Applying the
Bayesian edge selection of \S\ref{sec:sampler} to the section, with the marginal estimates
$(\hat\kappa_j, \hat\alpha_j, \hat\tau_j, \hat\sigma)$ of the fit above held fixed,
the posterior concentrates sharply, retaining or pruning nearly every one of the $305$
candidates decisively. The selection is also informative beyond the screen: the edge from
\textit{Sox10} to the myelin gene \textit{Sec14l5}, admitted by marginal partial
correlation, is pruned once \textit{Mog} and \textit{Cldn11} enter the model, correctly
reading \textit{Sox10} as conditionally independent of \textit{Sec14l5} given the myelin
genes it regulates \citep{stolt2002sox10}. The posterior summary, the retained graph, and
prior-sensitivity and internal-consistency checks are in Section~\ref{supp:data} of the
Supplementary Material.

\begin{figure}[htbp]
\centering
\includegraphics[width=\textwidth]{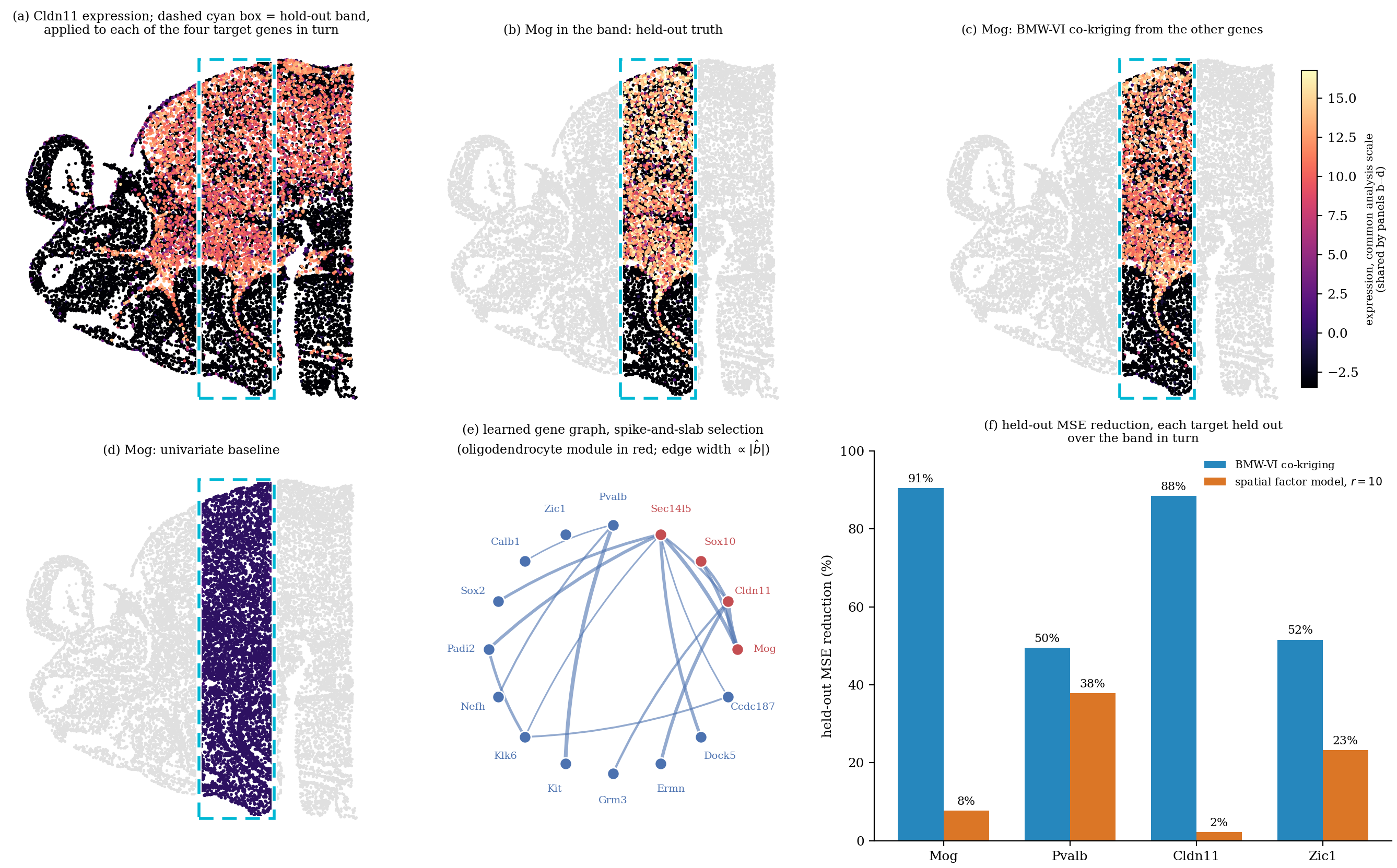}
\caption{One cerebellar MERFISH section ($n = 19{,}809$ cells, $p = 1{,}122$ genes). (a)
Expression of the myelin gene \textit{Cldn11}; the dashed cyan box marks the hold-out band, which
is applied to each of the four target genes in turn. (b)--(d) The example target
\textit{Mog}, coloured by its expression on the common analysis scale with one colour bar
shared by the three panels: (b) the held-out truth, (c) the multivariate co-kriging from the
other genes, and (d) the univariate baseline, which has no observation of \textit{Mog} in
the band. (e) The learned gene graph from the spike-and-slab selection, with the
oligodendrocyte module in red. (f) Held-out mean squared error reduction for each target
gene, each held out over the band in turn, for BMW-VI co-kriging and for the rank-$10$
spatial factor model.}
\label{fig:merfish}
\end{figure}

\section{Discussion}\label{sec:discussion}

The process proposed here extends the graphical Gaussian process family to arbitrary
spatial graphs, with the variable graph learned from the data. Several directions follow.
The triangular construction fixes a variable ordering. The eigenmode likelihood of
\S\ref{sec:graphlearn} scores any ordering at unchanged cost, so sampling or averaging over
orderings, as in \citet{jin2024bag}, is a natural next step. A related extension is to
update the spatial parameters and the nugget jointly with the graph rather than fixing them
at their preliminary estimates. On the computational side, learned preconditioners and
graphics hardware would extend the fine-mesh regime
\citep{gardner2018gpytorch,wang2019exact}. Posterior contraction theory for the learned
graph as $p$ grows would connect to the spectral convergence rates of
\citet{calder2022improved}. The application analysed a single tissue section. The atlas
provides $147$ further sections of the same brain and three additional specimens
\citep{zhang2023molecular}, so the reproducibility of the learned graph and of the
prediction gains can be examined directly.

\section*{Acknowledgement}
D. Dey was supported by Texas A\&M University start-up funds. Portions of this research
were conducted with the advanced computing resources provided by Texas A\&M Department of
Statistics Arseven Computing Cluster.

\section*{Supplementary material}
The Supplementary Material includes proofs of the theoretical results, computational
details, complete simulation and application results, and code reproducing all analyses.
A Python implementation, \texttt{bmwspatial}, is available at
\url{https://github.com/Ddey07/bmwspatial}.

\bibliographystyle{biometrika}
\bibliography{bwmp}

\appendix
\renewcommand{\thesection}{S\arabic{section}}
\renewcommand{\thetable}{S\arabic{table}}
\renewcommand{\thefigure}{S\arabic{figure}}
\renewcommand{\theequation}{S\arabic{equation}}
\renewcommand{\thetheorem}{S\arabic{theorem}}
\renewcommand{\theproposition}{S\arabic{proposition}}
\renewcommand{\thelemma}{S\arabic{lemma}}
\renewcommand{\thecorollary}{S\arabic{corollary}}
\renewcommand{\theremark}{S\arabic{remark}}
\renewcommand{\thealgocf}{S\arabic{algocf}}
\setcounter{section}{0}
\setcounter{table}{0}
\setcounter{figure}{0}
\setcounter{equation}{0}
\setcounter{theorem}{0}
\setcounter{proposition}{0}
\setcounter{lemma}{0}
\setcounter{corollary}{0}
\setcounter{remark}{0}
\setcounter{algocf}{0}

\renewcommand{\theHsection}{S.\arabic{section}}
\renewcommand{\theHtable}{S.\arabic{table}}
\renewcommand{\theHfigure}{S.\arabic{figure}}
\renewcommand{\theHequation}{S.\arabic{equation}}
\renewcommand{\theHtheorem}{S.\arabic{theorem}}
\renewcommand{\theHproposition}{S.\arabic{proposition}}
\renewcommand{\theHlemma}{S.\arabic{lemma}}
\renewcommand{\theHcorollary}{S.\arabic{corollary}}
\renewcommand{\theHremark}{S.\arabic{remark}}
\renewcommand{\theHalgocf}{S.\arabic{algocf}}

\section*{Supplementary Material}

This Supplementary Material contains the proofs of all results in the main paper
(Section~\ref{supp:proofs}), the Krylov approximation theory used by the matrix-free
computations (Section~\ref{supp:krylov}), the complexity table and measured timings
(Section~\ref{supp:scaling}), the settings and acceptance ratios of the
variable-graph sampler (Section~\ref{supp:sampler}),
complete simulation results for every setting of Table~1 of the main paper
(Section~\ref{supp:sims}), and additional results for the spatial-transcriptomics application
(Section~\ref{supp:data}).

\section{Proofs of the results in the main paper}\label{supp:proofs}

\begin{proof}[of Lemma~\ref{lem:exist}]
The operator $\bL$ is block lower-triangular with diagonal blocks $(\kappa_i^2 I + \Le)^{\alpha_i/2}$,
so $\det \bL = \prod_{i=1}^{p} \prod_{h=1}^{n} (\kappa_i^2 + \lambda_h)^{\alpha_i/2}$. Since
$\Le \succeq 0$ and $\kappa_i > 0$, every factor satisfies $\kappa_i^2 + \lambda_h \ge \kappa_i^2 > 0$,
so $\bL$ is invertible; the shift $\kappa_i^2$ is what regularizes the null eigenvalue
$\lambda_1 = 0$ of $\Le$. Hence $w = \bL^{-1} z$ is Gaussian with covariance
$\bL^{-1} D_\tau \bL^{-\T} \succ 0$ and precision $\Qw = \bL^{\T} D_\tau^{-1} \bL \succ 0$. As
$\det \bL$ is a product of diagonal blocks it does not involve the off-diagonal coefficients, so
invertibility and propriety hold for every $\{b_{ij}\}$ respecting the ordering. On $\R^d$ or a
manifold, with $\Le$ replaced by $-\Delta$, $w = \bL^{-1} z$ with $z$ white noise is a bounded
linear image of a white-noise measure; its finite-dimensional laws are therefore
Kolmogorov-consistent, and $w_i \in L^2$ whenever $\nu_i = \alpha_i - d/2 > 0$, the Whittle--Mat\'ern
existence condition.
\end{proof}

\begin{proof}[of Theorem~\ref{thm:ci}]
Write $\bL_{mi}$ for the $(m,i)$ block of $\bL$. Here
$\bL_{ii} = (\kappa_i^2 I + \Le)^{\alpha_i/2}$ is symmetric positive definite,
$\bL_{mi} = b_{mi} I$ if $(m, i) \in \GV$, and $\bL_{mi} = 0$ otherwise; triangularity gives
$\bL_{mi} = 0$ for $m < i$. Writing $D_\tau^{-1} = \diag(\tau_1^{-2},\ldots,\tau_p^{-2})\otimes I_n$,
the $(i,j)$ block of $\Qw = \bL^{\T} D_\tau^{-1}\bL$ is
$\sum_{m} \tau_m^{-2}\,\bL_{mi}^{\T}\bL_{mj}$. Take $j < i$ without loss of generality. Non-zero
contributions can arise from two sources only. The term $m = i$ contributes
$\tau_i^{-2} b_{ij}(\kappa_i^2 I + \Le)^{\alpha_i/2}$ when $(i,j) \in \GV$. A term $m > i$ contributes
$\tau_m^{-2} b_{mi} b_{mj} I$ when $m$ is a common child of $i$ and $j$. The term $m = j$ vanishes
because $\bL_{ji} = 0$. Hence
\begin{equation}\label{eq:Qblock}
(\Qw)_{ij} = \tau_i^{-2} b_{ij}(\kappa_i^2 I + \Le)^{\alpha_i/2}
+ \Bigl(\sum_{m:\, (m,i), (m,j) \in \GV} \tau_m^{-2} b_{mi} b_{mj}\Bigr) I .
\end{equation}
Because every $\tau_m^{-2} > 0$, the innovation variances rescale each contribution by a
positive constant and create no cancellation across the two terms. If $(i,j) \notin \GV$ and no
common child exists, both terms vanish identically for all
parameter values. Conversely, if $(i,j) \in \GV$, the first term is a non-zero multiple of a
positive-definite matrix whenever $b_{ij} \neq 0$; if a common child exists, the second term
is non-zero on an open set of cross-dependence coefficient values. The zero pattern of $\Qw$ is therefore that of
the moral graph of $\GV$. Since $w$ is Gaussian with precision $\Qw$, a zero block of the
precision is equivalent to conditional independence of $w_i$ and $w_j$ given the remaining
variables, jointly over locations.

Triangularity is essential. For a general operator symbol $H(\lambda)$, the spectral
precision is $H^{*} H$, and its $(i,j)$ entry $\sum_m \bar H_{mi} H_{mj}$ mixes all rows $m$.
A zero entry of $H$ therefore does not imply a zero entry of $H^{*}H$. The triangular
structure, together with the block-diagonal innovation covariance $D_\tau^{-1}$, restricts the
sum to the two structured terms of \eqref{eq:Qblock}.
\end{proof}

\begin{proof}[of Proposition~\ref{prop:ggp}]
(i) is immediate from Theorem~\ref{thm:ci}: that theorem shows the $(i,j)$ block of $\Qw$ vanishes
identically, for every parameter value, exactly when $(i,j) \notin \GV^{m}$, and for a Gaussian
vector a zero block of the precision is conditional independence of $w_i$ and $w_j$ given the
remaining variables, which is Definition~1 of \citet{dey2022graphical} for $\GV^{m}$. A graph with
no immoralities adds no moral edges, so there $\GV^{m}$ is the skeleton of $\GV$.

(ii) Expand each field in the eigenbasis of $\Le$. A diagonal block
$(\kappa_i^2 I + \Le)^{\alpha_i/2}$ acts on the $h$th eigencomponent as multiplication by
$(\kappa_i^2 + \lambda_h)^{\alpha_i/2}$, and the off-diagonal blocks $b_{ij} I$ act componentwise,
so \eqref{eq:gen} decouples across $h$ into the $p$-variate systems $L_h\, w^{(h)} = z^{(h)}$ with
$w^{(h)} = (u_h^{\T} w_1, \ldots, u_h^{\T} w_p)^{\T}$, $z^{(h)} \sim N(0, D_{\tau,0})$ and
$L_h = D_h + B$, $D_h = \diag\{(\kappa_i^2 + \lambda_h)^{\alpha_i/2}\}$. Hence
$w^{(h)} \sim N(0, \Sigma_h)$ with $\Sigma_h = L_h^{-1} D_{\tau,0} L_h^{-\T}$, and the marginal
spectral density of $w_i$ is $S_i(\lambda_h) = [\Sigma_h]_{ii}$. If $i$ has no parents then the
$i$th equation reads $(\kappa_i^2 I + \Le)^{\alpha_i/2} w_i = z_i$ with $z_i \sim N(0, \tau_i^2 I)$,
so $w_i \sim N\{0, \tau_i^2(\kappa_i^2 I + \Le)^{-\alpha_i}\}$. In general $L(\lambda)^{-1}$ is lower
triangular with
\[
g_{ij}(\lambda) = [L(\lambda)^{-1}]_{ij}
= \sum_{P:\, j \rightsquigarrow i} (-1)^{|P|}
\frac{\prod_{(v \leftarrow u) \in P} b_{vu}}{\prod_{v \in P} (\kappa_v^2 + \lambda)^{\alpha_v/2}},
\]
the sum over directed paths $P$ from $j$ to $i$, which is non-zero only for
$j \in \{i\} \cup \mathrm{an}(i)$. Since $D_{\tau,0}$ is diagonal,
$S_i(\lambda) = \sum_j \tau_j^2\, g_{ij}(\lambda)^2$. Every ancestral term contains the factor
$(\kappa_i^2 + \lambda)^{-\alpha_i}$ carried by $g_{ii}$ multiplied by a further decaying density,
so it is subleading; hence $S_i(\lambda) = \tau_i^2(\kappa_i^2 + \lambda)^{-\alpha_i}\{1 + o(1)\}$ as
$\lambda \to \infty$, and the high-frequency behaviour, and thus the smoothness
$\nu_i = \alpha_i - d/2$, is that of variable $i$ alone.
\end{proof}

\begin{proof}[of Theorem~\ref{thm:logdet}]
The operator $\bL$ is block lower-triangular, so $\det \bL = \prod_{j=1}^{p} \det \bL_{jj}$;
the off-diagonal blocks, which carry all cross-dependence coefficients, do not enter. By the spectral calculus,
$\det \bL_{jj} = \prod_{h=1}^{n} (\kappa_j^2 + \lambda_h)^{\alpha_j/2}$. Taking logarithms
and summing over $j$ gives the first display. Since $\Qw = \bL^{\T} D_\tau^{-1}\bL$ and $D_\tau$
is block-diagonal with $\det D_\tau = \prod_j \tau_j^{2n}$,
$\log|\Qw| = 2\log|\!\det\bL| + \log|D_\tau^{-1}| = 2\log|\!\det\bL| - 2n\sum_{j=1}^p \log\tau_j$.
The expression depends on $\theta$ only
through $(\kappa_j, \alpha_j, \tau_j)$, evaluated against the spectrum of $\Le$, which does not
depend on the data or on $\theta$; the innovation variances contribute only the separable term
$-2n\sum_j\log\tau_j$, and the cross-dependence coefficients do not appear.
\end{proof}

\begin{proof}[of Theorem~\ref{thm:nugget} and Corollary~\ref{cor:ceiling}]
Let $\Le = V \Lambda V^{\T}$ and $z = V^{\T} y$. Both covariance terms are functions of
$\Le$, hence diagonalized by $V$, so the $z_h$ are independent with
$z_h \sim N(0, d_h)$, where $d_h = s(\lambda_h) + \sigma^2$ and
$s(\lambda) = \tau^2(\kappa^2 + \lambda)^{-\alpha}$. For independent Gaussians with variances
$d_h(\theta)$, the Fisher information for a variance parameter $\theta$ is
$2^{-1}\sum_h d_h^{-2} (\partial d_h/\partial\theta)^2$. With $\theta = \sigma^2$ we have
$\partial d_h/\partial\sigma^2 = 1$ for every $h$, so
$\mathcal{I}(\sigma^2) = (2\sigma^4)^{-1} \sum_h \pi_h^2$ with $\pi_h = \sigma^2/d_h$. The
Cram\'er--Rao bound gives $\mathrm{var}(\hat\sigma^2) \ge 2\sigma^4/\neff$ for any unbiased
estimator.

For the corollary, monotonicity of $s$ gives $s(\lambda_h) \ge s_{\min} > 0$ for every $h$,
so $\pi_h \le \sigma^2/(s_{\min} + \sigma^2)$ uniformly and
$\neff/n \le \{\sigma^2/(s_{\min} + \sigma^2)\}^2$, a bound free of $n$. Under in-fill
sampling the normalized graph spectrum remains within $[0, \lambda_{\max}]$, so the bound
persists as $n \to \infty$. Since each $\pi_h > 0$, $\neff \to \infty$ and consistency is
unaffected, but the efficiency $\neff/n$ never approaches one. If
$\sigma^2 \ll s_{\min}$, the ceiling is at most $(\sigma^2/s_{\min})^2$, and the relative
standard deviation of any unbiased estimator exceeds
$\{2 n^{-1} (s_{\min}/\sigma^2)^{2}\}^{1/2}$, which remains large unless $n$ exceeds order
$(s_{\min}/\sigma^2)^2$. The continuum Whittle--Mat\'ern operator has unbounded spectrum, so
for any $\sigma^2 > 0$ the proportion of modes with $s(\lambda_h) < \sigma^2$ tends to one
and $\neff$ has exact order $n$. This recovers the strong identifiability regime of
\citet{tang2021identifiability}, in which moreover the smoothness is held fixed; here
$\alpha$ is free, which is the non-microergodic direction.
\end{proof}

\section{Krylov approximation, Gauss quadrature and trace estimation}\label{supp:krylov}

This section states and proves the numerical results used in \S\ref{sec:comp} of the main
paper. Throughout,
$A \in \R^{n\times n}$ is symmetric with $A = U\Lambda U^{\T}$, and
$f(A) = U f(\Lambda) U^{\T}$; in the paper $A = \Le$. Given $v \neq 0$, the Lanczos process
sets $q_1 = v/\|v\|$ and iterates, for $j = 1, \ldots, m$, with $\beta_0 q_0 = 0$,
\begin{equation}\label{eq:lanczos}
h = A q_j, \quad \alpha_j = q_j^{\T} h, \quad
\tilde h = h - \alpha_j q_j - \beta_{j-1} q_{j-1}, \quad
\beta_j = \|\tilde h\|, \quad q_{j+1} = \tilde h/\beta_j .
\end{equation}
In exact arithmetic $q_1, \ldots, q_m$ form an orthonormal basis of the Krylov subspace
$\mathcal{K}_m(A, v) = \mathrm{span}(v, Av, \ldots, A^{m-1}v)$. With
$Q_m = (q_1, \ldots, q_m)$ and the tridiagonal matrix $T_m$ of the coefficients,
\begin{equation}\label{eq:lanczosmat}
A Q_m = Q_m T_m + \beta_m q_{m+1} e_m^{\T}, \qquad Q_m^{\T} A Q_m = T_m ,
\end{equation}
at the cost of $m$ matrix--vector products.

\begin{proposition}\label{prop:exact}
For $0 \le k \le m-1$, $A^k v = \|v\| Q_m T_m^k e_1$. Consequently
$q(A)v = \|v\| Q_m q(T_m) e_1$ for every polynomial $q$ of degree at most $m - 1$.
\end{proposition}

\begin{proof}
We use induction on $k$. For $k = 0$, $v = \|v\| Q_m e_1$. Because $T_m$ is tridiagonal,
$T_m^k e_1$ is supported on its first $k+1$ coordinates, so $e_m^{\T} T_m^k e_1 = 0$ for
$k \le m - 2$. Then \eqref{eq:lanczosmat} gives
$A^{k+1} v = \|v\| (Q_m T_m + \beta_m q_{m+1} e_m^{\T}) T_m^k e_1
= \|v\| Q_m T_m^{k+1} e_1$, the residual term vanishing by the support property. Linearity
extends the identity to polynomials.
\end{proof}

The Lanczos approximation of a matrix function is
$f(A) v \approx f_m = \|v\| Q_m f(T_m) e_1$, with $f(T_m)$ computed from the small
eigendecomposition of $T_m$. Equivalently, $f_m = q^{*}(A) v$ where $q^{*}$ interpolates $f$
at the eigenvalues of $T_m$, the Ritz values.

\begin{theorem}\label{thm:kryerr}
Suppose the spectrum of $A$ lies in $[a, b]$, and let
$E_{m-1}(f) = \min_{\deg q \le m-1} \max_{a \le \lambda \le b} |f(\lambda) - q(\lambda)|$.
Then $\| f(A) v - f_m \| \le 2 \|v\| E_{m-1}(f)$.
\end{theorem}

\begin{proof}
For any polynomial $q$ of degree at most $m-1$, Proposition~\ref{prop:exact} gives
$f(A)v - f_m = \{f(A) - q(A)\}v - \|v\| Q_m \{f(T_m) - q(T_m)\} e_1$. The first term is
bounded in norm by $\max_{[a,b]}|f - q| \|v\|$ by the spectral theorem. The eigenvalues of
$T_m$ lie in $[a, b]$ by the Courant--Fischer characterization, and $Q_m$ has orthonormal
columns, so the second term obeys the same bound. Adding the bounds and minimizing over $q$
completes the proof.
\end{proof}

For $f$ analytic in a neighbourhood of $[a,b]$, such as
$f(\lambda) = (\kappa^2 + \lambda)^{\alpha/2}$ with $\kappa > 0$, Bernstein's theorem yields
geometric decay of $E_{m-1}(f)$ in $m$. This is why small Krylov dimensions suffice in
practice.

\begin{lemma}\label{lem:shift}
For any real $c$, $\mathcal{K}_m(A + cI, v) = \mathcal{K}_m(A, v)$. The Lanczos process on
$A + cI$ produces the same basis $Q_m$, and its tridiagonal matrix is $T_m + cI$.
\end{lemma}

\begin{proof}
The binomial expansion gives
$(A + cI)^j v \in \mathrm{span}(v, \ldots, A^j v)$, so the subspaces coincide and the
orthonormalization \eqref{eq:lanczos} produces the same vectors. The diagonal coefficients
shift by $c$, and each $\beta_j$ is unchanged because the residual $\tilde h$ is invariant
under the shift.
\end{proof}

Taking $c = \kappa_j^2$ shows that a single Lanczos run on $\Le$ serves every variable's
fractional operator, inverse and derivative; only the scalar evaluations at the Ritz values
change.

For traces, fix $v$ and define the spectral measure of the pair $(A, v)$ by
$\rho(\lambda) = \sum_h \mu_h^2 1\{\lambda_h \le \lambda\}$ with $\mu_h = u_h^{\T} v$, so
that $v^{\T} f(A) v = \int f \, d\rho$.

\begin{theorem}\label{thm:gauss}
Let $T_m$ be the Lanczos tridiagonal matrix for $(A, v)$, with eigenpairs $(\theta_r, s_r)$.
The rule with nodes $\theta_r$ and weights $\omega_r = \|v\|^2 (e_1^{\T} s_r)^2$
approximates $\int f d\rho$ by $\sum_{r=1}^m \omega_r f(\theta_r)$, and it is the $m$-point
Gauss rule for $\rho$: it is exact for all polynomials of degree at most $2m - 1$.
\end{theorem}

\begin{proof}
The monic orthogonal polynomials of $\rho$ satisfy a three-term recurrence whose coefficients
are exactly the Lanczos coefficients, since $q_{j+1} = \psi_j(A) v / \|\psi_j(A) v\|$;
substituting this expression into \eqref{eq:lanczos} reproduces the recurrence of the
$\psi_j$. Hence $T_m$ is the Jacobi matrix of $\rho$. Its eigenvalues are the roots of
$\psi_m$, which are the Gauss nodes, and the weights $\omega_r$ are the Christoffel numbers by
the Golub--Welsch construction \citep{golub2010matrices}. Exactness to degree $2m-1$ is the
defining property of Gauss quadrature.
\end{proof}

\begin{proposition}\label{prop:hutch}
Let $z$ have independent entries with $E(z_i) = 0$ and $E(z_i^2) = 1$. Then
$E(z^{\T} M z) = \tr (M)$ for any square $M$. For Rademacher $z$ and symmetric $M$,
$\mathrm{var}(z^{\T} M z) = 2\sum_{i \ne j} M_{ij}^2$, which is minimal among
independent-entry probe distributions.
\end{proposition}

\begin{proof}
The first claim follows from $E(z_i z_j) = 1\{i=j\}$. For the second, Rademacher entries
satisfy $z_i^2 = 1$, so the diagonal contributes no variance; expanding the quadratic form
leaves $2 \sum_{i \neq j} M_{ij}^2$. A Gaussian probe has variance twice the squared
Frobenius norm, which is no smaller \citep{hutchinson1990stochastic}.
\end{proof}

Estimating each $z_\ell^{\T} f(A) z_\ell$ by Theorem~\ref{thm:gauss} and averaging over
$m_{\mathrm{H}}$ Rademacher probes gives the stochastic Lanczos quadrature estimator of
$\tr\{ f(A)\}$ \citep{ubaru2017fast}. Its error is the sum of a quadrature part, of order
$E_{2m-1}(f)$ and negligible for smooth $f$, and a Monte Carlo part of order
$m_{\mathrm{H}}^{-1/2}$ \citep{chen2021analysis}. Applied with $A = \Le$ and
$f = \log(\kappa_j^2 + \cdot)$, this estimates the contribution of block $j$ to
Theorem~\ref{thm:logdet}. The pooled nodes and weights do not depend on $\theta$, so the same
quadrature serves the log-determinant and its derivatives, for every variable, at every
optimization step.

\subsection{Score components for the exact refinement}

The score identity of the main paper follows from
$E_{w\mid y}(w^{\T}Ax) = \mu^{\T}A\mu + \tr(A \Qpost^{-1})$. The component derivatives
are as follows. With $\Qw = \bL^{\T} D_\tau^{-1}\bL$
and, for a parameter $\vartheta \in \{\kappa_i, \alpha_i, b_{ij}\}$ entering $\bL$,
$\partial_\vartheta \Qw = (\partial_\vartheta \bL)^{\T} D_\tau^{-1}\bL + \bL^{\T} D_\tau^{-1}(\partial_\vartheta \bL)$,
the component derivatives are
\begin{align*}
\partial_{\kappa_i}\bL &: \ \text{diagonal block } i, \quad
\alpha_i \kappa_i (\kappa_i^2 I + \Le)^{\alpha_i/2 - 1},
& \partial_{\kappa_i}\log|\Qw| &= 2\alpha_i \kappa_i {\textstyle\sum_h} (\kappa_i^2 + \lambda_h)^{-1}, \\
\partial_{\alpha_i}\bL &: \ \text{diagonal block } i, \quad
\tfrac{1}{2}\log(\kappa_i^2 I + \Le)(\kappa_i^2 I + \Le)^{\alpha_i/2},
& \partial_{\alpha_i}\log|\Qw| &= {\textstyle\sum_h} \log(\kappa_i^2 + \lambda_h), \\
\partial_{b_{ij}}\bL &: \ \text{block } (i,j) \text{ equal to } I,
& \partial_{b_{ij}}\log|\Qw| &= 0,
\end{align*}
where the $\log|\Qw|$ derivatives are unchanged by the innovation variances, which enter
$\log|\Qw|$ only through the separable term $-2n\sum_j\log\tau_j$. The innovation variances act
on $\Qw$ through $D_\tau^{-1}$ rather than through $\bL$; writing $(\bL w)_i$ for the $i$th
spatial block of $\bL w$, so that $\partial_{\log\tau_i}\Qw = -2\tau_i^{-2}\,\bL^{\T}(e_i e_i^{\T}\otimes I)\bL$
and $\partial_{\log\tau_i}\log|\Qw| = -2n$, the score is
\[
\partial_{\log\tau_i}\ell = -n + \tau_i^{-2}\bigl\{\|(\bL\mu)_i\|^2
+ \tr[\bL^{\T}(e_i e_i^{\T}\otimes I)\bL\,\Qpost^{-1}]\bigr\},
\]
with stationary point the update $\tau_i^2 = n^{-1} E_{w\mid y}\|(\bL w)_i\|^2$, computed from
the same conjugate-gradient and stochastic-trace solves used for the other parameters. The
nugget score is $\partial_\sigma \ell = -M/\sigma + \sigma^{-3}\{\|y-\mu\|_{\Ocal}^2 +
\tr(D_{\Ocal} \Qpost^{-1})\}$.
Every term is evaluated with the conjugate-gradient and stochastic-trace machinery described
in the main paper, and the trace solves are shared across all parameters within a step.

\section{Complexity and measured scaling}\label{supp:scaling}

Table~\ref{tab:complexity} gives the time complexity of each computational step and
separates the scaling in $n$ from the scaling in $p$. Every step of estimation and
prediction is linear in both, up to the constants $m$, $m_{\mathrm{H}}$ and
$T_{\mathrm{CG}}$. The only step with a higher cost in $p$ is the dense form of the
graph-learning likelihood, and its sparse form is near-linear for a bounded-degree variable
graph. Table~\ref{tab:timing} reports measured times for the variational fit. All local
timings in this Supplementary Material were measured on an Apple M4 laptop with ten cores,
four performance and six efficiency, and $16$ gigabytes of unified memory.
Graph construction and the one-time quadrature setup are at or below a tenth of a second
across the small grid, and $5$ and $21$ seconds at $n = 1{,}200{,}000$. Step times grow in
proportion to $np$ across the grid. At $n = 1{,}200{,}000$ the operation count is unchanged
but the Lanczos basis no longer fits in cache, so memory traffic raises the step time to
roughly eight times the linear extrapolation from the smaller sizes; the proportionality in
$p$ at moderate $n$ is unaffected, as the per-setting times in Table~\ref{stab:timing} and
the application fit of \S6 of the main paper show. The measured per-iteration
comparison with the graphical Mat\'ern is in Section~\ref{supp:highp}.

\begin{table}[htbp]
\tbl{Time complexity of each computational step, with its separate scaling in the number of
locations $n$ and the number of variables $p$}{%
\begin{tabular}{llll}
Step & Cost & In $n$ & In $p$ \\[3pt]
Nearest-neighbour graph and $\Le$ & $O(n \log n)$ & $n\log n$ & none \\
Quadrature setup & $O(m_{\mathrm{H}} m\, \mathrm{nnz})$ & linear & none \\
$\log|\Qw|$ and its derivatives & $O(p\, m\, m_{\mathrm{H}})$ & none & linear \\
Operator application & $O(m\, \mathrm{nnz}\, p + |E_V| n)$ & linear & linear \\
Variational step & $O(m\, \mathrm{nnz}\, p)$ & linear & linear \\
Conjugate-gradient solve & $O(T_{\mathrm{CG}}\, m\, \mathrm{nnz}\, p)$ & linear & linear \\
Refinement step & $O(m_{\mathrm{H}} T_{\mathrm{CG}}\, m\, \mathrm{nnz}\, p)$ & linear & linear \\
Prediction, mean and variance & $O\{(1{+}m_{\mathrm{var}}) T_{\mathrm{CG}}\, m\, \mathrm{nnz}\, p\}$ & linear & linear \\
Graph likelihood, dense (\S\ref{sec:eigenlik}) & $O(n p^3)$ & linear & cubic \\
Graph likelihood, sparse (\S\ref{sec:eigenlik}) & $O(n p\, \bar{d}^{\,2})$ & linear & near-linear \\
Dense eigendecomposition of $\Le$ (alternative) & $O(n^3)$ & cubic & none \\
Dense joint likelihood (naive benchmark) & $O(n^3 p^3)$ & cubic & cubic \\
\end{tabular}}
\label{tab:complexity}
\begin{tabnote}
$\mathrm{nnz}$, number of non-zeros of $\Le$, of order $n$; $m$, Krylov dimension;
$m_{\mathrm{H}}$, number of probe vectors; $m_{\mathrm{var}}$, probes for predictive
variances; $T_{\mathrm{CG}}$, conjugate-gradient iterations; $\bar d$, maximum in-degree of the
variable graph. Memory is $O(np)$ throughout. The dense graph likelihood is cubic in $p$ only
because the nugget densifies each per-mode covariance; the sparse form removes this, leaving a
cost near-linear in $p$ for a bounded-degree variable graph. The dense eigendecomposition is
of the $n\times n$ spatial operator $\Le$ alone and so is independent of $p$; the naive
benchmark that factorizes the full $np\times np$ covariance is instead cubic in both, at
$O(n^3 p^3)$, and is the cost the matrix-free scheme is designed to avoid.
\end{tabnote}
\end{table}

\begin{table}[htbp]
\tbl{Measured times for the BMW-VI fit on an Apple M4 laptop, ten cores and $16$
gigabytes of unified memory, at Krylov dimension $15$}{%
\begin{tabular}{rrrrrrr}
$n$ & $p$ & $np$ & Graph (s) & Spectrum (s) & Per step (ms) & Full fit (s) \\[3pt]
10000 & 3 & $3\times10^4$ & 0.02 & 0.10 & 36 & 11 \\
4000 & 10 & $4\times10^4$ & 0.01 & 0.05 & 27 & 8 \\
1000 & 100 & $10^5$ & 0.00 & 0.02 & 45 & 13 \\
4000 & 100 & $4\times10^5$ & 0.01 & 0.05 & 156 & 47 \\
1200000 & 3 & $3.6\times10^6$ & 5.33 & 21.4 & 11043 & 3340 \\
\end{tabular}}
\label{tab:timing}
\begin{tabnote}
Graph, construction of the nearest-neighbour graph and operator $\Le$; Spectrum, the one-time
stochastic Lanczos quadrature setup; Per step, mean time of one variational step over twenty
steps; Full fit, end-to-end time of a complete fit, taken as the setup plus $300$ steps. At
$n = 1{,}200{,}000$ the step is dominated by memory traffic rather than floating-point
work, and exceeds a linear extrapolation from the smaller sizes by roughly a factor of
eight.
\end{tabnote}
\end{table}

\section{The variable-graph sampler: settings and acceptance ratios}\label{supp:sampler}

Algorithm~\ref{alg:pt} was run with temperatures $1, 1.55, 2.4, 3.7,
5.8, 9.0$, move probabilities $0.30$, $0.42$ and $0.28$ for the
birth-death, swap and random-walk moves, prior inclusion probability $\pi = 0.2$, slab
standard deviation $\tau = 0.25$, random-walk step $\varrho = 0.04$, and two independent
tempered runs. Inclusion probabilities discard a burn-in of $300$ iterations. The eigenmode
likelihood \eqref{eq:eigenlik} was validated against the dense evaluation to $10^{-8}$. The
move weights are proposal tuning parameters and the sampler is not sensitive to them; the
higher weight on the swap move reflects its role in bridging confounded graphs.

\begin{algorithm}[htbp]
\caption{Parallel-tempered reversible-jump spike-and-slab sampler for the variable graph.}
\label{alg:pt}
\KwIn{spectral data $\{z_h\}$; candidate edges $E$; prior $(\pi, \tau)$; inverse temperatures
$\beta_1 = 1 > \cdots > \beta_K$}
\For{each sweep}{
  \For{each replica $r = 1, \ldots, K$, with tempered log-likelihood $\beta_r \ell$}{
    with probability $0.3$: birth or death; pick $e \in E$; if inactive, propose
    $\gamma_e = 1$, $b_e \sim N(0, \tau^2)$; accept with
    $\log R = \beta_r \Delta\ell + \log\{\pi/(1-\pi)\}$; if active, propose removal with the
    reverse ratio\;
    else with probability $0.42$: edge swap; remove a uniformly chosen active edge and
    add a uniformly chosen inactive edge with $b_e \sim N(0, \tau^2)$; accept with
    $\log R = \beta_r \Delta\ell$\;
    else: random-walk update of one active $b_e$\;
  }
  \For{adjacent replica pairs $(r, r{+}1)$}{
    swap full states with probability
    $\min[1, \exp\{(\beta_r - \beta_{r+1})(\ell_{r+1} - \ell_r)\}]$\;
  }
  record the inclusion indicators of the replica with $\beta_1 = 1$\;
}
\KwOut{edge-inclusion probabilities and the median-probability graph}
\end{algorithm}

\subsection{Acceptance ratios}

\subsection*{Acceptance ratios for the variable-graph sampler}

Write $\gamma \in \{0,1\}^{|E|}$ for the edge-inclusion indicators and
$b = \{b_e : \gamma_e = 1\}$ for the active cross-dependence coefficients. At inverse
temperature $\beta_r$ the sampler targets
\[
\pi_{\beta_r}(\gamma, b) \; \propto \; \exp\{\beta_r\, \ell(\gamma, b)\}
\prod_{e \in E} \pi^{\gamma_e}(1-\pi)^{1-\gamma_e}
\prod_{e:\, \gamma_e = 1} N(b_e;\, 0, \tau^2),
\]
where $\ell$ is the eigenmode log-likelihood \eqref{eq:eigenlik} and $N(\,\cdot\,;0,\tau^2)$ is
the slab density. Each move is a Metropolis--Hastings--Green step \citep{green1995reversible};
write $\Delta\ell$ for the change in $\ell$ it induces and accept it with probability
$\min(1, e^{\log R})$. Because the slab proposal is drawn from the slab prior, the two cancel
in the Green ratio, which is what produces the simple forms used in Algorithm~\ref{alg:pt}.

A \emph{birth} activates an inactive edge $e$ with $b_e \sim N(0,\tau^2)$. The proposal density
of $b_e$ equals its slab prior and the Jacobian is one, so
\[
\log R_{\mathrm{birth}} = \beta_r \Delta\ell + \log\frac{\pi}{1-\pi},
\]
and its reverse \emph{death} has
$\log R_{\mathrm{death}} = \beta_r \Delta\ell + \log\{(1-\pi)/\pi\}$. An \emph{edge swap}
removes an active edge $e_1$ and activates an inactive edge $e_2$ with
$b_{e_2} \sim N(0,\tau^2)$. The active-edge count is preserved, so the two Bernoulli factors
cancel; the added edge's slab prior cancels its proposal density and likewise for the removed
edge; and the uniform choices of $e_1$ among active and $e_2$ among inactive edges cancel,
leaving
\[
\log R_{\mathrm{swap}} = \beta_r \Delta\ell.
\]
A \emph{coefficient update} is a symmetric random walk $b_e \to b_e + \eta$,
$\eta \sim N(0, \varrho^2)$, on an active edge; the proposal cancels and only the likelihood and
the slab prior remain,
\[
\log R_{\mathrm{rw}} = \beta_r \Delta\ell
+ \frac{1}{2\tau^2}\{ b_e^2 - (b_e + \eta)^2 \}.
\]
A \emph{replica exchange} proposes swapping the states of adjacent replicas $r$ and $r+1$,
accepted with
\[
\log R_{\mathrm{exch}} = (\beta_r - \beta_{r+1})(\ell_{r+1} - \ell_r).
\]
Every $\Delta\ell$ is a difference of the eigenmode log-likelihood \eqref{eq:eigenlik}: by
Theorem~\ref{thm:logdet} its prior log-determinant is common to all graphs and cancels, and the
remaining per-mode determinants are the small $p\times p$ sparse factorizations of
\S\ref{sec:eigenlik}, so no determinant over the $np$-dimensional space is ever formed.

\section{Complete simulation results}\label{supp:sims}

\subsection{Recovery and prediction in every setting}

Tables~\ref{stab:recov} and~\ref{stab:pred} report, for every setting of Table~1 of the main
paper, recovery of the identified parameters and held-out prediction with interval coverage.
Figure~\ref{sfig:identified8} shows the recovery boxplots for all eight fifty-replicate
settings with strong cross-dependence, of which the main text displays four; recovery is
uniformly high and the recovery slopes are near one. Figure~\ref{sfig:cokrig} displays the
structured-hold-out comparison summarized in \S\ref{sec:predsim} of the main paper. The
variational nugget sits at its optimization floor in every setting and the refinement moves
it toward the truth, as \S\ref{sec:nugget} of the main paper predicts
(Table~\ref{tab:signr}).

\begin{figure}[htbp]
\centering
\includegraphics[width=\textwidth]{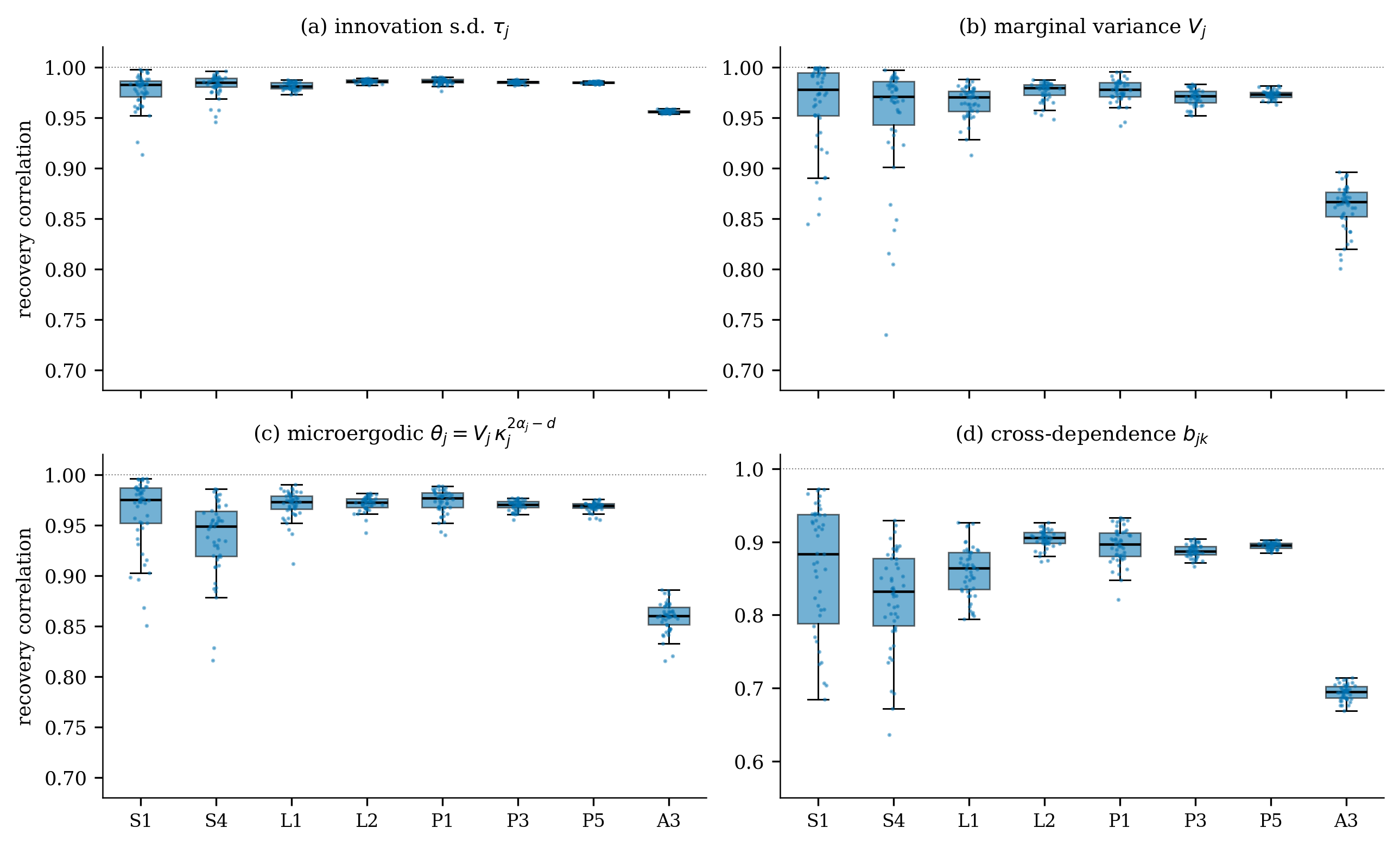}
\caption{Recovery of the identified parameters by BMW-ML in all eight
fifty-replicate settings with strong cross-dependence, from $n = 500$, $p = 8$ (S1) to the
stress shape $p = 2{,}000$ at $n = 300$ (A3). Each box holds fifty per-replicate
correlations between estimate and truth; panels show $\tau_j$, $V_j$, $\theta_j$ and
$b_{jk}$.}
\label{sfig:identified8}
\end{figure}

\begin{figure}[htbp]
\centering
\includegraphics[width=\textwidth]{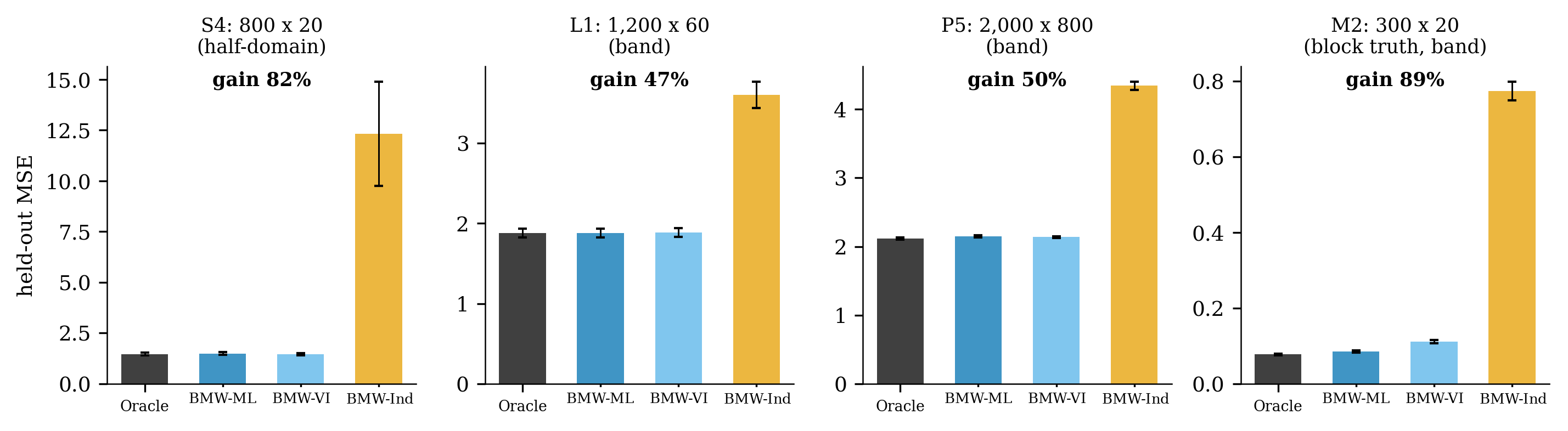}
\caption{Structured-hold-out mean squared error for the oracle, the two-stage estimator
(BMW-ML), its first stage (BMW-VI) and the graph-free fit (BMW-Ind), in four settings.
Annotations give the mean
per-replicate reduction of BMW-ML relative to BMW-Ind; error bars are two Monte Carlo
standard errors. The fitted model matches the oracle throughout, and the value of the variable
graph appears where prediction must borrow across variables.}
\label{sfig:cokrig}
\end{figure}
The recovery correlations of the weakly identified individual parameters are lower and
ridge-dominated, as expected: pooled over the well-specified settings, the BMW-ML
correlations for $\kappa_j$ and $\alpha_j$ lie well below those of the identified functions
$\tau_j$, $V_j$ and $\theta_j$, which is the empirical face of the in-fill confounding
discussed in \S\ref{sec:nugget} of the main paper.

Figure~\ref{sfig:bcoupling} isolates the effect of cross-dependence strength on recovery
of the cross-dependence coefficients, within each scale tier, so strength is not confounded
with scale. At $n = 500$, $p = 8$ the recovery correlation falls steeply as the
cross-dependence weakens. At the larger settings the effect washes out. With
$np \ge 7\times10^4$ even the mixed regime is recovered above $0.85$, because the
coefficients borrow strength across all $p$ surfaces.

\begin{figure}[htbp]
\centering
\includegraphics[width=\textwidth]{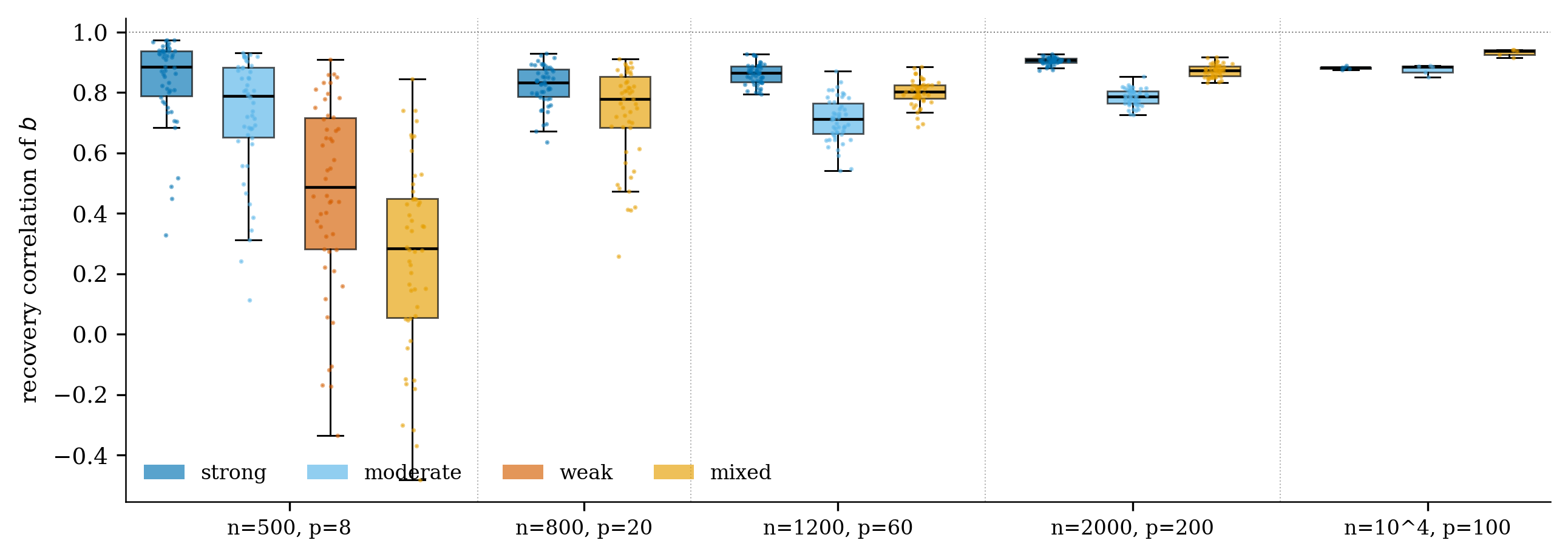}
\caption{Recovery correlation of the cross-dependence coefficients $b$ by strength of the
cross-dependence, grouped by scale tier. Boxes are per-replicate correlations of
BMW-ML. Colours denote the strong, moderate, weak and mixed regimes. All
settings have fifty replicates except the $n = 10^4$ tier, which has five. The strength
effect dominates at the smallest setting and vanishes as the dimensions grow.}
\label{sfig:bcoupling}
\end{figure}

\begin{table}[htbp]
\tbl{Recovery of the identified parameters by BMW-ML, all well-specified settings: mean per-replicate correlation between estimate and truth, with Monte Carlo standard errors in parentheses}{%
\begin{tabular}{lrrcccc}
Setting & $n$ & $p$ & $\tau$ & $V$ & $\theta$ & $b$ \\[3pt]
S1 & 500 & 8 & 0.98 (0.00) & 0.96 (0.01) & 0.96 (0.00) & 0.84 (0.02) \\
S2 & 500 & 8 & 0.98 (0.00) & 0.96 (0.01) & 0.96 (0.01) & 0.72 (0.03) \\
S3 & 500 & 8 & 0.98 (0.00) & 0.96 (0.01) & 0.96 (0.01) & 0.46 (0.05) \\
Smix & 500 & 8 & 0.97 (0.00) & 0.96 (0.01) & 0.96 (0.01) & 0.26 (0.04) \\
S4 & 800 & 20 & 0.98 (0.00) & 0.95 (0.01) & 0.94 (0.01) & 0.82 (0.01) \\
S4mix & 800 & 20 & 0.98 (0.00) & 0.95 (0.01) & 0.93 (0.01) & 0.73 (0.02) \\
L1s & 1,200 & 60 & 0.98 (0.00) & 0.97 (0.00) & 0.97 (0.00) & 0.86 (0.00) \\
L1m & 1,200 & 60 & 0.98 (0.00) & 0.97 (0.00) & 0.97 (0.00) & 0.71 (0.01) \\
L1x & 1,200 & 60 & 0.98 (0.00) & 0.97 (0.00) & 0.97 (0.00) & 0.80 (0.01) \\
L2s & 2,000 & 200 & 0.99 (0.00) & 0.98 (0.00) & 0.97 (0.00) & 0.90 (0.00) \\
L2m & 2,000 & 200 & 0.99 (0.00) & 0.97 (0.00) & 0.97 (0.00) & 0.78 (0.00) \\
L2x & 2,000 & 200 & 0.99 (0.00) & 0.97 (0.00) & 0.97 (0.00) & 0.87 (0.00) \\
P1 & 2,000 & 50 & 0.99 (0.00) & 0.98 (0.00) & 0.97 (0.00) & 0.90 (0.00) \\
P2 & 2,000 & 150 & 0.99 (0.00) & 0.98 (0.00) & 0.97 (0.00) & 0.89 (0.00) \\
P3 & 2,000 & 300 & 0.99 (0.00) & 0.97 (0.00) & 0.97 (0.00) & 0.89 (0.00) \\
P4 & 2,000 & 500 & 0.98 (0.00) & 0.97 (0.00) & 0.97 (0.00) & 0.89 (0.00) \\
P5 & 2,000 & 800 & 0.98 (0.00) & 0.97 (0.00) & 0.97 (0.00) & 0.90 (0.00) \\
A1 & 2,000 & 300 & 0.99 (0.00) & 0.97 (0.00) & 0.97 (0.00) & 0.89 (0.00) \\
A2 & 1,000 & 600 & 0.99 (0.00) & 0.95 (0.00) & 0.94 (0.00) & 0.86 (0.00) \\
A3 & 300 & 2000 & 0.96 (0.00) & 0.86 (0.00) & 0.86 (0.00) & 0.69 (0.00) \\
L3s & 10,000 & 100 & 0.97 (0.00) & 0.99 (0.00) & 0.98 (0.00) & 0.88 (0.00) \\
L3m & 10,000 & 100 & 0.98 (0.00) & 0.99 (0.00) & 0.99 (0.00) & 0.88 (0.01) \\
L3x & 10,000 & 100 & 0.98 (0.00) & 0.99 (0.00) & 0.98 (0.00) & 0.93 (0.01) \\
\end{tabular}}
\label{stab:recov}
\begin{tabnote}
Fifty replicates per setting (five for L3s--L3x). $V_j$ and $\theta_j$ are computed from the estimated $(\kappa_j,\alpha_j,\tau_j)$ through the operator spectrum.
\end{tabnote}\end{table}

\begin{table}[htbp]
\tbl{Held-out prediction across all settings: root mean squared error on the primary misaligned hold-out, 90\% interval coverage of BMW-ML, and the co-kriging gain on the structured hold-out}{%
\begin{tabular}{lccccccc}
Setting & Oracle & GM & BMW-Ind & BMW-VI & BMW-ML & Coverage & Structured gain (\%) \\[3pt]
S1 & 0.935 & 1.231 & 0.976 & 0.938 & 0.939 & 0.91 & 53.0 (4.8) \\
S2 & 0.943 & 1.190 & 0.963 & 0.945 & 0.946 & 0.90 & 17.6 (4.0) \\
S3 & 0.950 & 1.187 & 0.957 & 0.952 & 0.953 & 0.90 & -1.0 (2.8) \\
Smix & 0.951 & 1.175 & 0.956 & 0.953 & 0.954 & 0.90 & -1.5 (2.3) \\
S4 & 1.146 & -- & 1.163 & 1.148 & 1.148 & 0.91 & 81.9 (1.5) \\
S4mix & 1.151 & -- & 1.158 & 1.152 & 1.152 & 0.91 & 82.9 (1.4) \\
M1 & 0.265 & 0.283 & 0.328 & 0.325 & 0.321 & 0.93 & 8.6 (3.1) \\
M1b & 0.270 & 0.277 & 0.327 & 0.330 & 0.327 & 0.93 & -1.1 (1.7) \\
M2 & 0.326 & 1.601 & 0.819 & 0.370 & 0.346 & 0.79 & 88.9 (0.3) \\
L1s & 1.107 & -- & 1.125 & 1.108 & 1.108 & 0.90 & 46.9 (1.2) \\
L1m & 1.110 & -- & 1.117 & 1.111 & 1.111 & 0.89 & 28.5 (1.3) \\
L1x & 1.111 & -- & 1.118 & 1.112 & 1.112 & 0.89 & 34.0 (1.2) \\
L2s & 1.140 & -- & 1.156 & 1.141 & 1.141 & 0.90 & 36.1 (0.4) \\
L2m & 1.142 & -- & 1.148 & 1.143 & 1.143 & 0.90 & 14.8 (0.4) \\
L2x & 1.143 & -- & 1.149 & 1.143 & 1.143 & 0.90 & 17.6 (0.4) \\
P1 & 1.132 & -- & 1.147 & 1.133 & 1.133 & 0.91 & 52.4 (1.3) \\
P2 & 1.122 & -- & 1.141 & 1.123 & 1.123 & 0.91 & 48.2 (0.8) \\
P3 & 1.148 & -- & 1.164 & 1.150 & 1.149 & 0.91 & 43.1 (0.4) \\
P4 & 1.150 & -- & 1.166 & 1.151 & 1.151 & 0.91 & 46.0 (0.4) \\
P5 & 1.150 & -- & 1.166 & 1.151 & 1.151 & 0.91 & 50.4 (0.3) \\
A1 & 1.148 & -- & 1.164 & 1.150 & 1.149 & 0.91 & 43.1 (0.4) \\
A2 & 1.161 & -- & 1.177 & 1.162 & 1.162 & 0.91 & 40.7 (0.4) \\
A3 & 1.163 & -- & 1.186 & 1.166 & 1.166 & 0.90 & 25.1 (0.4) \\
L3s & 1.123 & -- & 1.138 & 1.125 & 1.124 & 0.89 & 49.1 (3.0) \\
L3m & 1.124 & -- & 1.131 & 1.126 & 1.125 & 0.89 & 22.7 (2.1) \\
L3x & 1.125 & -- & 1.131 & 1.127 & 1.126 & 0.89 & 27.3 (1.8) \\
\end{tabular}}
\label{stab:pred}
\begin{tabnote}
GM, the graphical Mat\'ern; entries are blank where its $p n^3$ cost exceeds the feasibility cap. The oracle predicts with the generating parameters. Structured gain is the mean per-replicate reduction in held-out mean squared error of BMW-ML relative to BMW-Ind, with Monte Carlo standard error in parentheses.
\end{tabnote}\end{table}

\subsection{The nugget across signal-to-noise ratios}

\begin{table}[htbp]
\tbl{BMW-VI nugget estimation across signal-to-noise ratios, well-specified model with
full observation}{%
\begin{tabular}{lcccc}
Signal-to-noise ratio & 4 & 8 & 15 & 30 \\[3pt]
Ratio of estimate to truth & 0.56 & 1.11 & 2.08 & 4.17 \\
Monte Carlo standard error & $<$0.01 & $<$0.01 & $<$0.01 & 0.01 \\
\end{tabular}}
\label{tab:signr}
\begin{tabnote}
The variational estimate is nearly constant while the true nugget varies over a $7.5$-fold
range; see \S\ref{sec:nugget}.
\end{tabnote}
\end{table}

\subsection{Behaviour in high dimensions}\label{supp:highp}

Settings P1--P5 raise $p$ from $50$ to $800$ at $n = 2{,}000$. Tables~\ref{stab:recov}
and~\ref{stab:pred} show that recovery of the identified parameters and the structured
co-kriging gains are flat in $p$, while fit time grows linearly; Fig.~\ref{sfig:highp}
collects the measured evidence. For the per-iteration cost comparison referenced in
\S\ref{sec:comptiming} of the main paper: one likelihood evaluation of the graphical Mat\'ern, which
forms and factorizes dense $n \times n$ blocks, takes $1.0$ seconds at $n = 400$, $p = 40$,
against $12$ milliseconds for one variational step of the proposed method at the same size,
a factor of $84$, and the gap widens with $p$: the variational step is measured through
$p = 320$, where the graphical Mat\'ern is no longer practical.

\begin{figure}[htbp]
\centering
\includegraphics[width=\textwidth]{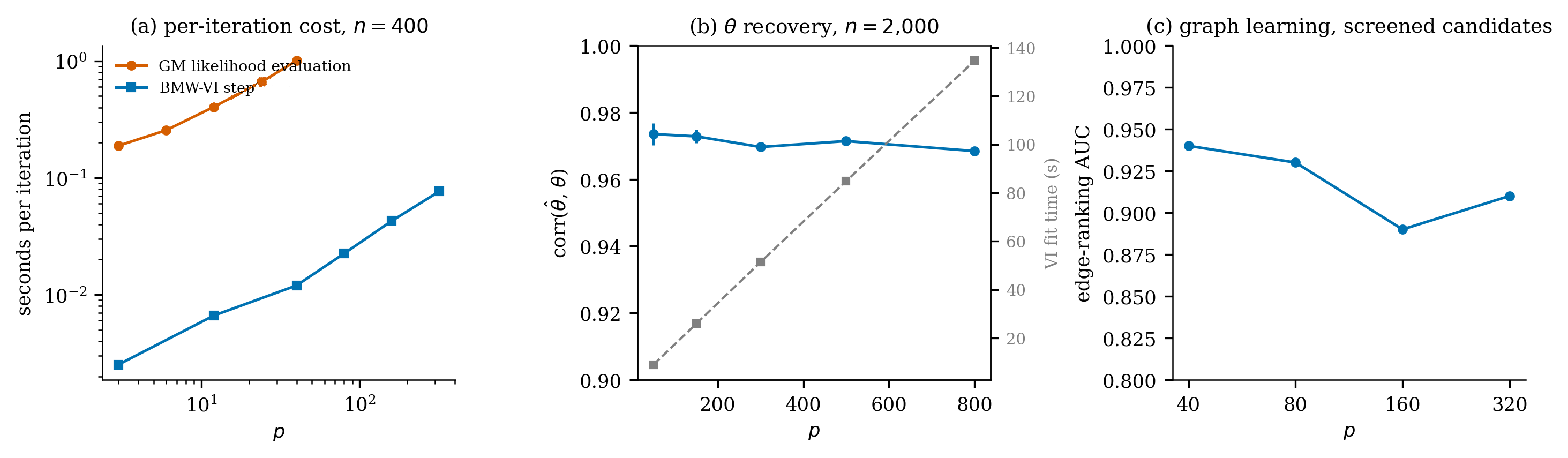}
\caption{Behaviour in high dimensions. (a) Measured per-iteration cost at $n = 400$: one
graphical Mat\'ern likelihood evaluation against one BMW-VI step. (b) Recovery of
the microergodic $\theta_j$ at $n = 2{,}000$ is flat in $p$ over fifty replicates, while the
variational fit time, right axis, grows linearly. (c) Edge-ranking area under the curve for
graph learning with a partial-correlation screen (setting G4); the dip at $p = 160$ tracks
the screen's recall.}
\label{sfig:highp}
\end{figure}

\subsection{Aspect-ratio study}

Settings A1--A3 fix $np = 6\times10^5$ and trade locations for variables, with shapes
$2{,}000\times300$, $1{,}000\times600$ and $300\times2{,}000$. Recovery of the identified
parameters degrades only at the most extreme shape: $\theta$ recovery is $0.97$, $0.95$ and
$0.86$ across the three shapes (Table~\ref{stab:recov}). Each variable contributes $n$
spectral observations to its own parameters, so it is $n$ and not $np$ that governs
per-variable recovery, while the cross-dependence coefficients continue to borrow strength
across the $p$ surfaces.

\subsection{Graph learning against known truth}\label{supp:graph}

Four settings assess the parallel-tempered spike-and-slab sampler. Each is summarized by
the separation of the posterior, the mean inclusion probability of true edges against that
of absent candidate pairs, together with the area under the edge-ranking curve. G1
generates from a block coregionalization with independent blocks: the mean probability of
within-block pairs is $0.75$ against $0.06$ for the $150$ across-block pairs, the largest
across-block probability is $0.24$, and the ranking area is $0.999$. G2 generates from a
Euclidean-distance graphical Mat\'ern, so the operator is misspecified: the separation is
$0.64$ for true edges against $0.13$ for false pairs, with area $0.901$. G3 varies the edge
density at $p = 40$: the mean probability of true edges rises from $0.40$ at $0.4p$ edges
to $0.58$ at $2p$, absent pairs stay pinned at the prior value $0.05$ throughout, and the
area lies between $0.92$ and $0.95$ across the in-degree-two and in-degree-three arms. G4
raises the dimension to $p = 40$, $80$, $160$ and $320$ with the partial-correlation screen
of the candidate set: true edges average $0.47$ to $0.66$ against $0.06$ to $0.14$ for
absent pairs, the areas are $0.94$, $0.93$, $0.89$ and $0.91$, and the dips track the
screen's recall of true edges rather than the sampler itself.

\begin{figure}[htbp]
\centering
\includegraphics[width=\textwidth]{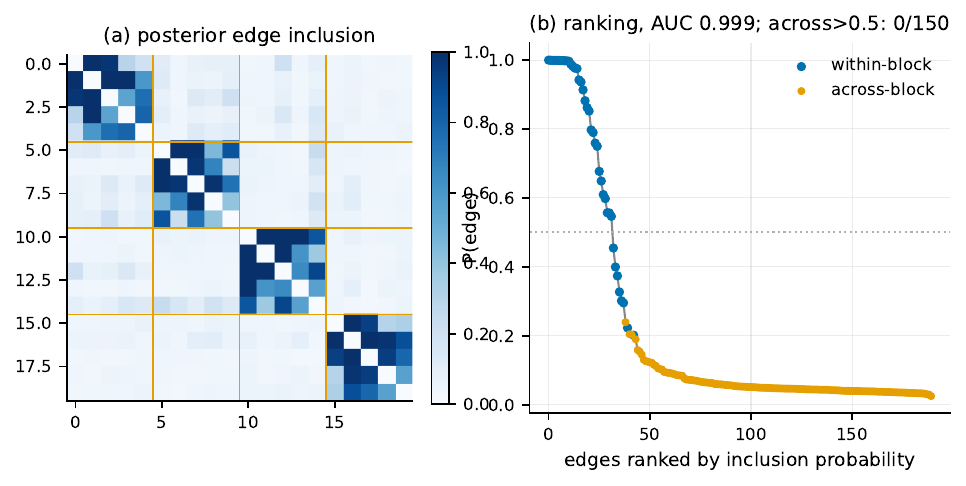}\\[4pt]
\includegraphics[width=\textwidth]{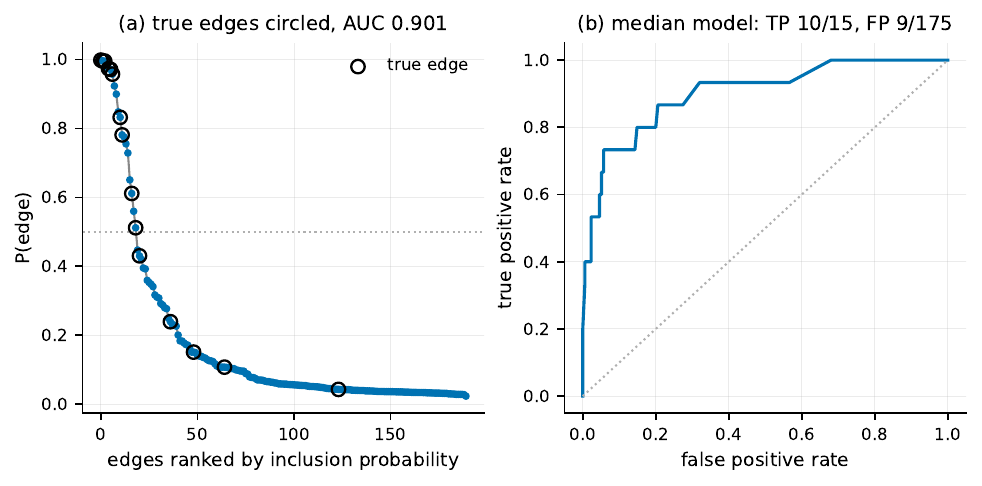}
\caption{Graph recovery under misspecification. Top: data generated from a
block-coregionalization model with no exact sparse representation (G1); (a) the posterior
edge-inclusion matrix recovers the block structure; (b) edges ranked by inclusion probability,
with within-block pairs cleanly separated from across-block pairs. Bottom: data generated from
a Euclidean-Mat\'ern graphical model, a spatial operator different from the one fitted (G2); (a)
true edges (circled) rank near the top of the posterior; (b) the corresponding ROC curve.}
\label{sfig:misspecgraph}
\end{figure}

\begin{figure}[htbp]
\centering
\includegraphics[width=\textwidth]{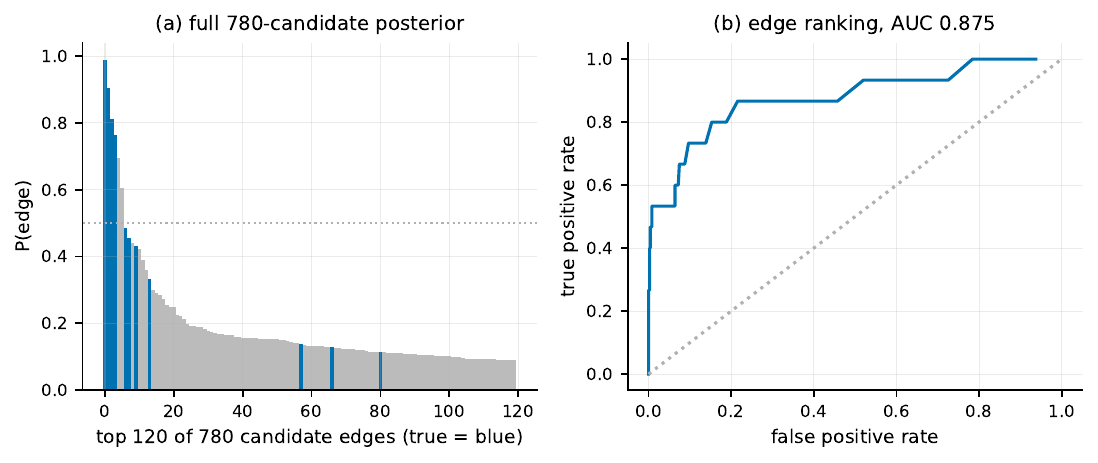}\\[4pt]
\includegraphics[width=\textwidth]{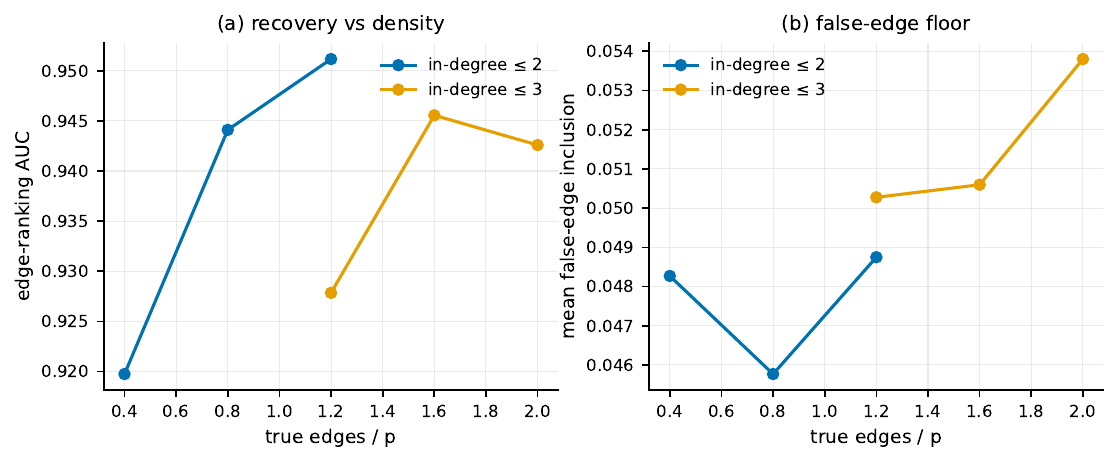}
\caption{Bayesian graph learning at $p = 40$. Top: the full $780$-candidate spike-and-slab
search (G3, sparse arm); (a) posterior inclusion probabilities of the top-ranked candidates,
true edges in blue; (b) the corresponding ROC curve. Bottom: recovery (a) and the mean
false-edge inclusion floor (b) as edge density increases, for candidate sets bounded to
in-degree $\le 2$ and $\le 3$ (G3, dense arms).}
\label{sfig:sslearn}
\end{figure}

\subsection{Computing environment}

All replicated simulations ran on a departmental Linux cluster, one replicate per task, four
cores and $16$ gigabytes of memory per task; every result file records its host, seeds and
per-method wall time. Table~\ref{stab:timing} reports the wall-clock cost of one replicate
of every setting on the Apple M4 laptop described above. The variational fit itself is fast and linear in
$np$; the totals are dominated by the evaluation harness, which fits every model twice, once
on the primary hold-out and once on the structured hold-out, and computes Hutchinson
predictive standard deviations for four methods. A passive convergence monitor recorded when
a windowed early-stopping criterion would have fired; within the fixed step budgets it fired
only at the L1 tier, on the final step, so the budgets are not over-generous.

\begin{table}[htbp]
\tbl{Wall-clock time for one replicate of every setting on a four-performance-core Apple M4 laptop: the BMW-VI fit, the BMW-ML refinement, and the complete replicate including the four predictions and the separately refitted structured hold-out}{%
\begin{tabular}{lrrrrrr}
Setting & $n$ & $p$ & $np$ & BMW-VI (s) & BMW-ML refinement (s) & Full replicate (s) \\[3pt]
S1 & 500 & 8 & $4,000$ & 2 & 18 & 527 \\
S2 & 500 & 8 & $4,000$ & 2 & 20 & 1,016 \\
S3 & 500 & 8 & $4,000$ & 2 & 21 & 657 \\
Smix & 500 & 8 & $4,000$ & 2 & 21 & 563 \\
S4 & 800 & 20 & $16,000$ & 3 & 63 & 272 \\
S4mix & 800 & 20 & $16,000$ & 3 & 64 & 281 \\
M1 & 500 & 5 & $2,500$ & 2 & 15 & 206 \\
M1b & 500 & 5 & $2,500$ & 2 & 16 & 221 \\
M2 & 300 & 20 & $6,000$ & 2 & 34 & 392 \\
L1s & 1,200 & 60 & $72,000$ & 11 & 232 & 1,078 \\
L1m & 1,200 & 60 & $72,000$ & 11 & 240 & 1,199 \\
L1x & 1,200 & 60 & $72,000$ & 11 & 242 & 1,206 \\
L2s & 2,000 & 200 & $4.0\times10^{5}$ & 39 & 1167 & 6,151 \\
L2m & 2,000 & 200 & $4.0\times10^{5}$ & 39 & 1227 & 6,286 \\
L2x & 2,000 & 200 & $4.0\times10^{5}$ & 39 & 1216 & 6,282 \\
P1 & 2,000 & 50 & $10^{5}$ & 9 & 273 & 1,406 \\
P2 & 2,000 & 150 & $3.0\times10^{5}$ & 26 & 856 & 3,940 \\
P3 & 2,000 & 300 & $6.0\times10^{5}$ & 51 & 1766 & 9,473 \\
P4 & 2,000 & 500 & $10^{6}$ & 85 & 2956 & 15,757 \\
P5 & 2,000 & 800 & $1.6\times10^{6}$ & 135 & 5061 & 25,416 \\
A1 & 2,000 & 300 & $6.0\times10^{5}$ & 51 & 1765 & 9,480 \\
A2 & 1,000 & 600 & $6.0\times10^{5}$ & 57 & 1983 & 10,284 \\
A3 & 300 & 2000 & $6.0\times10^{5}$ & 79 & 3027 & 12,563 \\
L3s & 10,000 & 100 & $10^{6}$ & 78 & 3307 & 15,054 \\
L3m & 10,000 & 100 & $10^{6}$ & 77 & 3446 & 15,831 \\
L3x & 10,000 & 100 & $10^{6}$ & 78 & 3446 & 15,808 \\
\end{tabular}}
\label{stab:timing}
\begin{tabnote}
The variational fit is linear in $np$: $9$, $26$, $51$, $85$ and $135$ seconds across P1--P5 as $np$ grows from $10^5$ to $1.6\times10^6$. The full replicate is dominated by the second, separately refitted structured hold-out and the Hutchinson predictive standard deviations, which cost several conjugate-gradient solves each. GM timing, where feasible, is included in the totals of the small settings.
\end{tabnote}\end{table}

\section{Additional results for the spatial-transcriptomics application}\label{supp:data}

\subsection*{Spatial factor model comparison}

The factor-model competitor of \S6 of the main paper was fitted on the identical
out-of-sample design: the same training cells, centring, common scale and hold-out band as
the main fit. Loadings are the leading $r$ principal components of the training expression
matrix, with held-out entries set to their centred value of zero, and each factor score
field is smoothed by a Vecchia--Mat\'ern Gaussian process with fifteen neighbours
\citep{vecchia1988estimation,guinness2018permutation}; a held-out gene is reconstructed
from the smoothed factors and its loading vector. This is the statistical core of the
factor methods used in spatial transcriptomics
\citep{velten2022mefisto,townes2023nonnegative}, though those implementations add
non-Gaussian likelihoods and nonnegativity, so we describe ours as a factor model in their
spirit rather than as either method. Table~\ref{stab:factor} reports the held-out
mean-squared-error reduction for ranks $10$, $20$ and $30$ against the same univariate
baseline as the main text. Rank $10$ is best throughout and deeper factorizations degrade,
because the panel is far from low-rank: the leading $10$, $20$ and $30$ components explain
only $40\%$, $47\%$ and $51\%$ of the training variance, and later components contribute
reconstruction noise for any single gene. The spatial smoothing is not the bottleneck: the
plug-in reconstruction from unsmoothed scores lands within two points of the smoothed
version for every target, so the deficit is attributable to the dense low-rank
cross-dependence itself.

\begin{table}[htbp]
\tbl{Held-out mean-squared-error reduction over the univariate baseline on the
transcriptomics band hold-out, in percent: the BMW-VI fit of the main text against the
rank-$r$ spatial factor model}{%
\begin{tabular}{lrrrr}
Target & BMW-VI & Factor $r=10$ & Factor $r=20$ & Factor $r=30$ \\[3pt]
\textit{Mog}    & $90.5$ & $7.7$  & $0.9$  & $-0.4$ \\
\textit{Cldn11} & $88.4$ & $2.2$  & $-4.3$ & $-5.3$ \\
\textit{Pvalb}  & $49.5$ & $37.8$ & $5.0$  & $-4.0$ \\
\textit{Zic1}   & $51.6$ & $23.3$ & $-2.7$ & $-5.8$ \\
\end{tabular}}
\label{stab:factor}
\begin{tabnote}
All methods use the same training cells, standardization and hold-out band; reductions are
$100\{1 - \mathrm{mse}/\mathrm{mse}_{\mathrm{uni}}\}$ against the univariate fit of the
main text. Negative values indicate error above the univariate baseline.
\end{tabnote}
\end{table}

\subsection*{Per-gene estimates and graph selection}

The per-gene parameter estimates for the full $1{,}122$-gene panel, the ranges,
smoothnesses, amplitudes and the implied marginal variances and microergodic parameters, are
released with the code as a machine-readable table. The candidate screen computes partial
correlations from a ridge-regularized inverse of the empirical gene--gene correlation matrix
across cells, with $10^{-2}$ added to the diagonal, and admits for each gene its at most
three strongest earlier-ordered partners with absolute partial correlation above $0.08$;
this yielded the $305$ candidates. The Bayesian edge selection over the
$305$ screened candidates uses the fitted per-gene $(\kappa_j, \alpha_j, \tau_j)$, with
expression centred and placed on the common scale of the main fit so the amplitudes are
preserved. The posterior concentrates sharply: $283$ candidates lie above inclusion
probability $0.9$ and twelve below $0.1$, so nearly every candidate is decisively retained
or pruned; the graph displayed in Fig.~\ref{sfig:merfishgraph} retains the $289$ candidates
above one half. Two checks support the selection. First, rerunning it with the prior inclusion
probability halved, from $0.25$ to $0.125$, alters the conclusion for five of the $305$
candidates. Second, the residual innovation amplitudes recovered inside the per-gene
regressions correlate at $0.987$ with the $\hat\tau_j$ of the main fit, so the selection
and the variational fit see the same amplitude model. Figure~\ref{sfig:merfishgraph}
summarizes the posterior, the screen-versus-posterior comparison with the explained-away
\textit{Sox10}\,$\to$\,\textit{Sec14l5} edge at inclusion probability $0.105$, and the
retained myelin and vascular modules, whose remaining edges are retained with probability
one.

\begin{figure}[htbp]
\centering
\includegraphics[width=\textwidth]{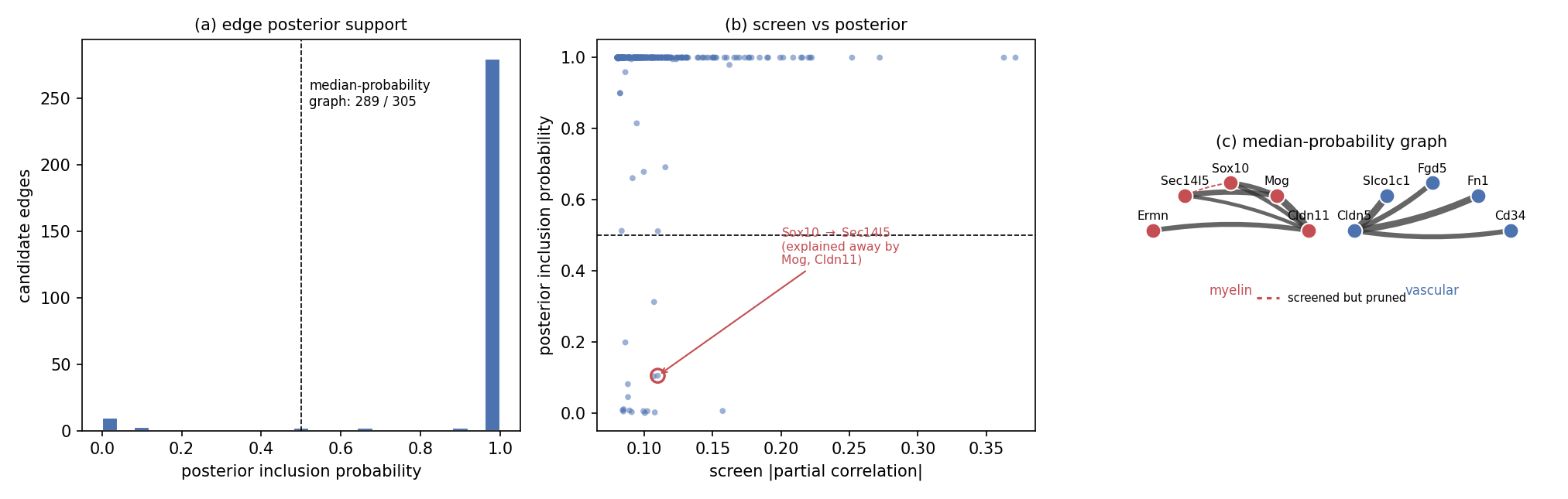}
\caption{Bayesian graph learning on the cerebellar section, over the screened candidate
edges. (a) Posterior edge-inclusion probabilities concentrate near zero and one. (b) Screen
strength against posterior support; the circled edge \textit{Sox10}\,$\to$\,\textit{Sec14l5}
is screened in yet pruned once the myelin genes are present. (c) Median-probability graph of
the oligodendrocyte-lineage (red) and vascular (blue) modules, edge width proportional to the
posterior cross-dependence.}
\label{sfig:merfishgraph}
\end{figure}

\end{document}